\documentclass[oneside,english]{amsart}

\usepackage{cmap}

\usepackage{concmath}
\usepackage[OT1]{fontenc}
\usepackage[latin9]{inputenc}
\usepackage{varioref}
\usepackage{amsthm}
\usepackage{amssymb}

\usepackage[usenames,dvipsnames]{xcolor}

\usepackage{cmap}
\usepackage[hyphens]{url}
\usepackage{hyperref}
\hypersetup{breaklinks=true}
\hypersetup{colorlinks=true,linkcolor=MidnightBlue,backref=page,citecolor=Magenta}
\usepackage{cleveref}

\usepackage[backend=bibtex,style=alphabetic,firstinits=true,maxbibnames=99,url=false]{biblatex}
\bibliography{VersionFinalHamiltonEqsGR.bib}

\usepackage{mathtools}

\usepackage[shortlabels]{enumitem}

\usepackage{tikz-cd}
\usetikzlibrary{decorations.pathreplacing, positioning}

\makeatletter
\numberwithin{equation}{section} 
\numberwithin{figure}{section} 
\theoremstyle{plain}
\theoremstyle{plain}
\newtheorem{thm}{Theorem}
  \theoremstyle{plain}
  \newtheorem{lem}{Lemma}
  \theoremstyle{plain}
  \newtheorem{prop}{Proposition}
  \theoremstyle{remark}
  \newtheorem*{note*}{Note}
  \theoremstyle{remark}
  \newtheorem*{conclusion*}{Conclusion}
  \theoremstyle{remark}
  \newtheorem{note}{Remark}
 \theoremstyle{definition}
  \newtheorem{example}{Example}
  \newtheorem{defc}{Definition}
  \theoremstyle{plain}
  \newtheorem{cor}{Corollary}

\usepackage{amsthm}
\usepackage{amsfonts}
\usepackage{amscd}
\usepackage[all]{xy}
\usepackage{pb-diagram,pb-xy}

\usepackage{tikz}
\usetikzlibrary{matrix,arrows}

\usepackage{tikz-cd}
\usetikzlibrary{decorations.pathreplacing, positioning}

\usepackage{soul}

\newcommand{\kf}{\mathfrak{k}}

\newcommand{\pf}{\mathfrak{p}}

\newcommand{\abs}[1]{\lvert#1\rvert}

\newcommand{\dif}{d}

\newcommand{\cI}{{\mathcal I}}
\newcommand{\cJ}{{\mathcal J}}
\newcommand{\cL}{{\mathcal L}}

\newcommand{\mR}{\mathbb{R}}

\usepackage[e]{esvect}

\usepackage{charter}

\makeatother

\makeatletter
\newcommand*\bigcdot{\mathpalette\bigcdot@{.5}}
\newcommand*\bigcdot@[2]{\mathbin{\vcenter{\hbox{\scalebox{#2}{$\m@th#1\bullet$}}}}}
\makeatother

\usepackage{todonotes}

\usepackage{babel}

\begin{document}

\title[Hamiltonian formalism for field theories and gravity]{A Hamiltonian formalism for general variational problems, with applications to first order gravity with basis}

\author{G. Quijón and S. Capriotti$^\dagger$}
\address{Departamento de Matem\'atica \\
Universidad Nacional del Sur and CONICET\\
 8000 Bah{\'\i}a Blanca \\
Argentina}
\email{$^\dagger$santiago.capriotti@uns.edu.ar}

\maketitle

\begin{abstract}
  The present article introduces a generalization of the (multisymplectic) Hamiltonian field theory for a Lagrangian density, allowing the formulation of this kind of field theories for variational problem of more general nature than those associated to a classical variational problem. It is achieved by realizing that the usual construction of the Hamiltonian equations can be performed without the use of the so called Hamiltonian section, whose existence is problematic when general variational problems are considered. The developed formalism is applied to obtain a novel multisymplectic Hamiltonian field theory for a first order formulation of gravity with basis in the full frame bundle. Also, aspects of a constraint algorithm in this generalized setting are discussed.
\end{abstract}

\tableofcontents

\section{Introduction}
\label{sec:introduction}

The Hamiltonian version of a field theory is an important step in the quantization of the system it is describing. Nevertheless, the usual construction requires the choice of a foliation of the spacetime with spacelike $3$-surfaces, and thus losing the explicit Lorenz symmetry of the theory. As a way to overcome this problem, efforts have been made to formulate Hamiltonian field theories from their Lagrangian counterpart preserving the basic symmetries of the theory \cite{W-1982,helein01hamiltonian,grabowska_lagrangian_2010,KRUPKOVA200293,EcheverriaEnriquez:2005ht,Echeverria-EnrIquez20007402}. On the other hand, it is often advantageous to work with variational problems in bundles not necessarily related to a jet bundle; this is a quite common need when symmetry reduction is considered \cite{castrillon_eulerpoincare_2007} or in alternate formulations for gravitation \cite{capriotti14:_differ_palat,Capriotti2020b}. These variational problems consist into finding extremals for the functional
\[
  A_\lambda\left(\sigma\right):=\int_M\sigma^*\lambda
\]
where $\lambda\in\Omega^m\left(W\right)$ is a $m$-form on a general bundle $\tau:W\to M^m$, and $\sigma:M\to W$ is a section of this bundle annihilating a set of forms on $W$. When $W=J^{k+1}\pi$, $\lambda=\cL$ and the set of forms is the contact structure on $J^{k+1}\pi$, this setting recovers the classical variational problem. 
It is reasonable to assume that not every variational problem should necessarily occur in a jet bundle and that, even when this is the case, the sections to be integrated into the action could not be holonomic.

In the present article, we propose a generalization of the construction of Hamiltonian field theories in the realm of general variational problems.

One of the problems one encounters when trying to perform this generalization, is that the usual formulation of Hamiltonian field theory out from its Lagrangian counterpart requires geometrical constructions that are intimately related to the jet bundle structures associated to the variational problem. In particular, the \emph{Hamiltonian form}, from which the Hamilton equations are written, requires for its definition of the notion of a \emph{Hamiltonian section}, which is a section of a bundle obtained through quotient from a multimomentum space that in turn is the space of affine maps on the jet bundle. It is quite evident that some of these structures are missing in general variational problems as those indicated above.

Nevertheless, it is possible to find geometrical structures that are common among variational problems in general. Between them, the more important for our purposes is the so called \emph{unified} or \emph{Lagrangian-Hamiltonian formalism} (also called \emph{Skinner-Rusk formalism}, after the authors who first proposed it in \cite{Skinner1983GeneralizedHD}). In the case of classical field theories, these formulations has been studied in several places; for example, see \cite{1751-8121-42-47-475207,375178,DeLeon:2002fc,2004JMP....45..360E}, and \cite{Vitagliano2010857} for a formulation using secondary calculus. Also, the reader can find a proposal for second order field theories avoiding the appearance of non symmetric momenta in \cite{Prieto-Martinez2015203}. It finds immediate application in the study of the Hamiltonian formulation for field theories with singular Lagrangians, where no method is available for such formulation. As we have been warned above, a crucial feature of the previously cited articles is that they work in the realm of the so called \emph{classical variational problems}; namely, the kind of variational problems on a jet bundle $J^{k+1}\pi$ (where $\pi:P\to M$ is the bundle of fields) prescribed by a functional determined through a Lagrangian density using the formula
\[
  A_\cL\left(s\right):=\int_M\left(j^{k+1}s\right)^*\cL.
\]
The crucial feature of the unified formalism is that it is possible to be generalized to variational problems more general than the classical variational problem. A first step in this direction was taken in the works \cite{doi:10.1142/S0219887818500445,Capriotti2020b}; in these articles, and using a construction described in \cite{GotayCartan}, it was possible to find an unified formalism for first order and Lovelock gravity. Also, a variation of this formalism was successfully employed in the definition of Routh reduction of a first order field theory by a general Lie group \cite{Capriotti2019}. In principle, the method used in these works can be interpreted as \emph{ad hoc}, since the relationship between such method and the canonical construction of the unified problem is not evident at first sight. One of the purposes of the present article is to prove that the construction of Gotay also replicates the usual unified formalism when the variational problem is classical. Although the basic scheme for this version of the unified formalism was outlined in the works previously cited, it is our conviction that it deserves its own exposition in a separate article. Additionally, and as a way to stress its practical value, another physical example (aside from the applications to gravity mentioned above) will be discussed, where the generalized construction of the unified formalism shows its adequacy in dealing with variational problems of general nature \cite{F_J_de_Urries_1998}.

Once we have established that the unified formalism can be extended to general variational problems, it remains to show how a Hamiltonian field theory can be extracted from it. The second aim of the present article is to describe a construction of the usual Hamiltonian form directly from the unified formalism, without recurring to a Hamiltonian section, and to generalize this procedure to general variational problems, in order to find a Hamiltonian version for them. This scheme is then applied to find a novel Hamiltonian formulation for gravity with basis, using for this end the description for this kind of gravity theory given in \cite{doi:10.1142/S0219887818500445}.

The final section of the article discusses some aspects of the constraint algorithm from this generalized viewpoint. The formulation of this algorithm for general variational problems is applied to some well-known examples.

\section{The unified formalism for higher order field theories}
\label{sec:class-lepage-equiv}

We will devote the next section to the discussion of an unified formalism for higher order field theories. Essentially, it is based in the notion of \emph{classical Lepage equivalent problem} as it is discussed in \cite{GotayCartan}, see also \cite{Krupka1986a,Krupka1986b}. The basic idea of this approach is to consider that, as in the Lagrangian formulation the sections of the jet bundle to be used in the functional should be holonomic, the contact structure of the jet bundle can be considered as a set of constraints, and so they can be put in the Lagrangian through Lagrange multipliers; from this viewpoint, the (multi)momenta are nothing but these multipliers. As there is a general way of performing this operation geometrically, these considerations free us from having to imagine how the space of multimomenta should be for each order.

\subsection{The classical Lepage equivalent problem for first order field theories}
\label{sec:class-lepage-equiv-1}

Let us briefly describe in this section the formalism developed by Gotay in \cite{GotayCartan,Gotay1991203}. First, suppose that we have a bundle $\tau:Q\to M$, and consider the \emph{contact subbundle} $I_{{\rm con},2}^m=I_{\rm con}\cap\Lambda^m_2 J^1\tau$ spanned by $m$-forms which are $2$-vertical and which, in view of the observations above, admits the following description:
\begin{equation}\label{eq:ContactFields}
  \left.I_{{\rm con},2}^m\right|_{j_x^1s}=\mathcal{L} \left\{\alpha\circ\theta\wedge\beta:\alpha\in T^*_{s\left(x\right)}Q, \beta\in\left(\Lambda^{m-1}_1J^1\tau\right)_{j_x^1s}\right\}\subset\wedge^m_2\left(J^1\tau\right),
\end{equation}
where
\[
  \left.\theta\right|_{j_x^1s}:=T_{j_x^1s}\tau_{10}-T_xs\circ T_{j_x^1s}\tau_1
\]
is the $V\tau$-valued contact $1$-form on $J^1\tau$. Also, the notation $\mathcal{L}\{\cdot \}$ denotes the linear span; in other words, an element $\rho$ in the contact subbundle $I_{{\rm con},2}^m$ is of the form
\begin{equation*}
 \rho=(\alpha_1\circ\theta)\wedge\beta_1+\dots+(\alpha_k\circ\theta)\wedge\beta_k,
\end{equation*}
sor some $k\in\mathbb{N}$ and with $\alpha_i,\beta_i$. We will call an element of $I_{{\rm con},2}^m$ with a single summand (i.e., $k=1$) a \emph{simple element}. Most of the proofs involving $I_{{\rm con},2}^m$ will be done for simple elements, since the case of arbitrary elements is similar.

Next, given a Lagrangian density ${\widehat{\cL}}$ considered as a $\tau_1$-horizontal $m$-form on $J^1\tau$, we construct the affine subbundle
\[
  W_{\widehat{\cL}}:={\widehat{\cL}}+I_{{\rm con},2}^m\subset\wedge^m_2\left(J^1\tau\right),
\]
which becomes a bundle $\tau_{\widehat{\cL}}:W_{\widehat{\cL}}\to J^1\tau$ through the restriction of the natural map $\tau_2^m:\wedge^m_2\left(J^1\tau\right)\to J^1\tau$. Because $\wedge^m_2\left(J^1\tau\right)$ has a canonical $m$-form, we can pullback it along the inclusion $W_{\widehat{\cL}}\hookrightarrow\wedge^m_2\left(J^1\tau\right)$ in order to have a canonical $m$-form $\lambda_{\widehat{\cL}}$ on $W_{\widehat{\cL}}$. Then we have the following characterization for the extremals of the Lagrangian field theory on the bundle $\tau:Q\to M$ associated to the Lagrangian density ${\widehat{\cL}}$.
\begin{prop}\label{prop:FieldTheoryEqsWL}
A section $s\colon U\subset M\rightarrow Q$ is critical for $(\tau\colon Q\rightarrow M,{\widehat{\cL}})$ if and only if there exists a section $\Gamma\colon U\subset M\rightarrow W_{{\widehat{\cL}}}$ such that
\begin{enumerate}[label={\arabic*)},nolistsep]
\item $\Gamma$ covers $s$, i.e. $\tau_{10}\circ\tau_{{\widehat{\cL}}}\circ\Gamma=s$, and
\item $\Gamma^*\left(X\lrcorner d\lambda_{{\widehat{\cL}}}\right)=0$, for all $X\in\mathfrak{X}^V(W_{{\widehat{\cL}}})$.
\end{enumerate}
\end{prop}
Here $\mathfrak{X}^{V}(W_{{\widehat{\cL}}})$ denotes the vector fields which are vertical w.r.t. the projection $W_{{\widehat{\cL}}}\to M$. $\Gamma$ is called a \emph{solution} of $(\tau\colon Q\rightarrow M,{\widehat{\cL}})$ or of $(W_{{\widehat{\cL}}},\lambda_{{\widehat{\cL}}})$.

\subsection{The unified formalism}
\label{sec:lepage-equiv-probl}

We will adapt this description in order to find the equations of motion for a higher order field theory given by a $\left(k+1\right)$-order Lagrangian density
\[
  \cL:J^{k+1}\pi\to\wedge^m\left(T^*M\right).
\]
Recall \cite{saunders89:_geomet_jet_bundl} that for every $k\geq1$ we have an immersion
\[
  \iota_{1,k}:J^{k+1}\pi\to J^1\pi_k.
\]
Then, the representation for this structure that we will use in the present article comes from the following result.
\begin{lem}
  The pullback form
  \[
    \theta_{k+1}:=\iota_{1,k}^*\theta,
  \]
  where $\theta\in\Omega^1\left(J^1\pi_k,V\pi_k\right)$ is the contact $1$-form on $J^1\pi_k$, induces the contact structure on $J^{k+1}\pi$.
\end{lem}
Let $\widehat{\cL}:J^1\pi_k\to\wedge^m\left(T^*M\right)$ be a Lagrangian density on $J^1\pi_k$ such that
\[
  \iota_{1,k}^*\widehat{\cL}=\cL;
\]
then, using the construction described in Section \ref{sec:class-lepage-equiv-1} for $\tau=\pi_k:J^k\pi\to M$, define the affine subbundle $W_{\widehat{\cL}}\subset\wedge^m_2\left(J^1\pi_k\right)$. It induces an affine subbundle $\widetilde{W_\cL}\subset\wedge^m_2\left(J^{k+1}\pi\right)$ through the formula
\[
  \left.\widetilde{W_\cL}\right|_{j_x^{k+1}s}:=\left\{\alpha\circ T_{j_x^{k+1}s}\iota_{1,k}:\alpha\in W_{\widehat{\cL}}\right\}.
\]
We are now ready for the introduction of the notion of \emph{unified formalism for higher order field theories}.
\begin{defc}[Unified formalism for $k+1$-order field theories]\label{def:k+1-unified}
  The \emph{unified formalism associated to the Lagrangian density $\cL:J^{k+1}\pi\to\wedge^m\left(T^*M\right)$} is the pair $\left(\widetilde{W_\cL},\Theta_\cL\right)$, where $\Theta_\cL$ is the pullback of the canonical $m$-form on $\wedge^m_2\left(J^{k+1}\pi\right)$ to $\widetilde{W_\cL}$. A section $\psi:U\subset M\to\widetilde{W_\cL}$ for the bundle ${\tau_\cL}:\widetilde{W_\cL}\to M$ is a solution for the \emph{Lagrangian-Hamiltonian problem posed by $\left(\widetilde{W_\cL},\Theta_\cL\right)$} is and only if
  \[
    \psi^*\left(X\lrcorner d\Theta_\cL\right)=0
  \]
  for every $X\in\mathfrak{X}^{V\widetilde{\tau_\cL}}\left(\widetilde{W_\cL}\right)$.
\end{defc}

\begin{note}
  It is customary in the literature to use the symbol $\Theta_\cL$ for the Cartan form associated to the Lagrangian density $\cL$; here we will deviate from this use, trusting that this choice of terminology will not cause any confusion, as the notion of a Cartan form will not be used in the present work. Therefore, the symbol $\Theta_\cL$ will indicate the $m$-form on $\widetilde{W_\cL}$ defined above.
\end{note}

\begin{note}
  It should be stressed that there are two different ways to define the $m$-form on $\widetilde{W_\cL}$: As described in Definition \ref{def:k+1-unified}, or using the pullback of the $m$-form on $W_{\widehat{\cL}}$ along the map $\iota_{1,k}$. It can be proved that both definitions give rise to the same form.
\end{note}

\subsubsection{Local expressions}
\label{sec:local-expressions}

Let us take a look at these constructions in local terms. Using Equation \eqref{eq:ContactFields} in the coordinates $\left(x^i,u^\alpha,u_I^\alpha\right),\abs{I}\leq k$ on $J^k\pi$, we have that $\rho\in W_{\widehat{\cL}}$ can be locally written as
\[
  \rho=\widehat{L}\eta+\sum_{i=1}^m\sum_{0\leq\abs{I}\leq k}p_\alpha^{I,i}\left(du^\alpha_I-\overline{u}_{I,j}dx^j\right)\wedge\eta_i.
\]
It means that an element  $\widetilde{\rho}\in\widetilde{W_\cL}$ becomes
\[
  \widetilde{\rho}={L}\eta+\sum_{i=1}^m\sum_{0\leq\abs{I}\leq k}p_\alpha^{I,i}\left(du^\alpha_I-{u}_{I+1_j}dx^j\right)\wedge\eta_i,
\]
where $\cL=L\eta$. Then we have coordinates $\left(x^i,u^\alpha,u_I^\alpha;p^{J,i}_\alpha\right)$, where $1\leq\abs{I}\leq k+1$ and $1\leq\abs{J}\leq k$ on $\widetilde{W_\cL}$, and in these coordinates the $m$-form $\Theta_\cL$ becomes simply
\begin{equation}\label{eq:LocalThetaL}
  \Theta_\cL={L}\eta+\sum_{i=1}^m\sum_{0\leq\abs{I}\leq k}p_\alpha^{I,i}\left(du^\alpha_I-{u}_{I+1_j}^\alpha dx^j\right)\wedge\eta_i.
\end{equation}
Please note that in the coordinates just defined, the velocities $u^\alpha_I$ are symmetric in the multiindex $I$, but the multimomentum coordinates $p_\alpha^{J,i}$ has mixed symmetry in these indices.

\subsection{Lagrangian formalism}
\label{sec:lagrangian-formalism}

We will need to relate the solutions for the Lagrangian-Hamiltonian problem with the solutions of the Euler-Lagrange equations for $\cL$. It will be done using the notion of bivariant Lepage equivalent problem from \cite{GotayCartan}. Namely, the Euler-Lagrange equations describe the extremals for the variational problem
\[
  A_\cL\left(s\right):=\int_M\left(j^{k+1}s\right)^*\cL,\qquad s:U\subset M\to P
\]
associated to the Lagrangian density $\cL$; on the other hand, the Lagrangian-Hamiltonian problem is the set of equations describing the extremals of the functional
\[
  A_{\Theta_\cL}\left(\sigma\right):=\int_M\sigma^*\Theta_\cL,\qquad\sigma:U\subset M\to\widetilde{W_\cL}.
\]
In principle, no relationship should exists between these extremals; our first result shows that every extremal of $A_{\Theta_\cL}$ projects along the map $\tau_\cL$ onto an extremal for $A_\cL$ (it is what is called \emph{covariance} of the variational problem posed by $A_{\Theta_\cL}$ in \cite{GotayCartan}).
\begin{lem}\label{lem:CovarianceUnified}
  If $\sigma:M\to\widetilde{W_\cL}$ is an extremal for $A_{\Theta_\cL}$, then there exists an extremal $s:M\to P$ for $A_\cL$ such that $\tau_\cL\circ\sigma=j^{k+1}s$.
\end{lem}
\begin{proof}[Proof (Gotay \cite{GotayCartan})]
  Let us define
  \[
    \phi:=\tau_\cL\circ\sigma:M\to J^{k+1}\pi.
  \]
  We will prove that $\phi^*\alpha=0$ for every contact form $\alpha$ on $J^{k+1}\pi$, and that it is an extremal for the functional $A_\cL$.

  First, let $\beta$ be a contact $m$-form on $J^{k+1}\pi$; then using the linear structure of $\wedge^m_2\left(J^{k+1}\pi\right)$ we can define the $\tau_\cL$-vertical vector field
  \begin{equation}\label{eq:TauLVertical}
    X_\beta\left(\rho\right):=\left.\frac{\text{d}}{\text{d}t}\right|_{t=0}\left(\rho+t\beta\left(\tau_\cL\left(\rho\right)\right)\right).
  \end{equation}
  Because $\Theta_\cL$ comes from a canonical $m$-form, we have that
  \begin{equation}\label{eq:TauLVerticalEOM}
    X_\beta\lrcorner d\Theta_\cL=\tau_\cL^*\beta,
  \end{equation}
  and so
  \[
    \phi^*\beta=\sigma^*\tau_\cL^*\beta=\sigma^*\left(X_\beta\lrcorner d\Theta_\cL\right)=0.
  \]
  It means that $\phi$ annihilates on every contact $m$-form; it is clear that $\phi^*\beta=0$ for contact $p$-forms with $p>k$. Now if $\beta$ is a contact $p$-form with $p<m$, then $\beta\wedge\alpha$ with $\alpha$ an arbitrary $\left(m-p\right)$-form is a contact $m$-form, and so $\phi^*\beta=0$ again. Thus there exists a section $s:M\to P$ such that
  \[
    \phi=j^{k+1}s.
  \]
  Now let prove that $s$ is an extremal for $A_\cL$. Consider a section $w:J^{k+1}\pi\to\widetilde{W_\cL}$ of $\tau_\cL$ such that
  \[
    w\circ j^{k+1}s=\sigma;
  \]
  then
  \[
    \sigma^*\Theta_\cL=\left(j^{k+1}s\right)^*w^*\Theta_\cL.
  \]
  But $\Theta_\cL$, being the pullback of the canonical $m$-form on $\wedge^m_2\left(J^1\pi_k\right)$ has the tautological property, namely
  \[
    w^*\Theta_\cL=w;
  \]
  furthermore, $w$ has its image in $\widetilde{W_\cL}$, and it means that
  \[
    w\equiv\cL\mod{\text{contact forms}}.
  \]
  Therefore
  \[
    \sigma^*\Theta_\cL=\left(j^{k+1}s\right)^*w^*\Theta_\cL=\left(j^{k+1}s\right)^*w=\left(j^{k+1}s\right)^*\cL,
  \]
  so
  \[
    A_{\Theta_\cL}\left(\sigma\right)=A_\cL\left(s\right),
  \]
  and $s$ should be an extremal for $A_\cL$, as required.
\end{proof}

Obviously, the correspondence between extremals can be achieved if we can assure that every extremal for $A_\cL$ can be lifted through $\tau_\cL$ to an extremal for $A_{\Theta_\cL}$; this property is described in \cite{GotayCartan} as \emph{contravariance}. 

\begin{lem}
  Every extremal for $A_\cL$ can be lifted to an extremal for $A_{\Theta_\cL}$. 
\end{lem}
\begin{proof}
  We will adapt the proof given in \cite{GotayCartan} to our purposes. We need to construct a section $w:J^{k+1}\pi\to\widetilde{W_\cL}$ for $\tau_\cL$ such that
  \[
    \sigma:=w\circ j^{k+1}s
  \]
  is extremal for $A_{\Theta_\cL}$, for every $s:M\to P$ extremal of $A_\cL$. Therefore, $\sigma$ is a solution for the equations
  \begin{equation}\label{eq:LagrangeHamiltonProblem2}
    \sigma^*\left(X\lrcorner d\Theta_\cL\right)=0,\qquad X\in\mathfrak{X}^{V\left(\pi_{k+1}\circ\tau_\cL\right)}\left(\widetilde{W_\cL}\right).
  \end{equation}
  Taking $X=X_\beta$ for $\beta$ a contact $m$-form (see proof of Lemma \ref{lem:CovarianceUnified}), we obtain that
  \[
    \sigma^*\left(X_\beta\lrcorner d\Theta_\cL\right)=\left(j^{k+1}s\right)^*w^*\tau_\cL^*\beta=0,
  \]
  so these part of the equations are verified for any $w$. Now we have the decomposition
  \[
    \left.T\widetilde{W_\cL}\right|_{w\left(J^{k+1}\pi\right)}=Tw\left(TJ^{k+1}\pi\right)\oplus V\tau_\cL,
  \]
  so that we can consider that $X$ in Equation \eqref{eq:LagrangeHamiltonProblem2} have the form
  \[
    H=Tw\circ V
  \]
  for some $V\in\mathfrak{X}^{V\pi_{k+1}}\left(J^{k+1}\pi\right)$; so
  \[
    0=\sigma^*\left(X\lrcorner d\Theta_\cL\right)=\left(j^{k+1}s\right)^*w^*\left(Tw\circ V\lrcorner d\Theta_\cL\right)=\left(j^{k+1}s\right)^*\left(V\lrcorner dw\right)
  \]
  where the tautological property for the canonical $m$-form on $\wedge^m_2\left(J^1\pi_k\right)$ was used. Therefore, it is the equation that section $w$ should obey in order to lift the extremal $s$ for $A_\cL$ to an extremal of $A_{\Theta_\cL}$.

  Let us write the equation
  \[
    \left(j^{k+1}s\right)^*\left(V\lrcorner dw\right)=0,\qquad V\in\mathfrak{X}^{V\pi_{k+1}}\left(J^{k+1}\pi\right)
  \]
  in local terms. From discussion carried out in Subsection \ref{sec:local-expressions}, we know that there will exist local functions $\lambda^{I,i}_\alpha\in C^{\infty}\left(J^{k+1}\pi\right)$ such that
  \[
    w={L}\eta+\lambda_\alpha^{I,i}\left(du^\alpha_I-{u}_{I+1_k}dx^k\right)\wedge\eta_i,
  \]
  and then these functions should obey the following system of differential equations
  \begin{align*}
    \lambda^{\left(I,i\right)}_\alpha&=\frac{\partial L}{\partial u^\alpha_{I+1_i}},\qquad\abs{I}=k,\\
    \left(\lambda^{\left(J,i\right)}_\alpha-\frac{\partial L}{\partial u^\alpha_{J+1_i}}\right)\eta&=d\lambda^{I+1_i,k}_\alpha\wedge\eta_k,\qquad0\leq\abs{J}<k,
  \end{align*}
  together with the equations of motion
  \[
    d\lambda^i_\alpha\wedge\eta_i-\frac{\partial L}{\partial u^\alpha}\wedge\eta=0,
  \]
  which are identically satisfied on an extremal $s:M\to P$ for the functional $A_\cL$. Therefore,
  \[
    \lambda^{I,i}_\alpha=
    \begin{cases}
      \displaystyle
      \frac{\partial L}{\partial u^\alpha_{I+1_i}}+c^{I,i}_\alpha,&\abs{I}=k,\\
      \displaystyle
      \frac{\partial L}{\partial u^\alpha_{I+1_i}}+D_k\lambda_\alpha^{I+1_i,k}+c^{I,i}_\alpha,&0\leq\abs{I}<k,
    \end{cases}
  \]
  where the arbitrary functions $c^{I,i}_\alpha$ are only constrained by the requirement of having zero symmetric part, namely $c^{\left(I,i\right)}_\alpha=0$, and $D_k$ is the total derivative operator. Thus, using these definitions we can construct the section $w$ that allows us to lift an extremal of $A_\cL$ to an extremal for $A_{\Theta_\cL}$.
\end{proof}

\begin{thm}
  The projection $\tau_\cL:\widetilde{W_\cL}\to J^{k+1}\pi$ establishes a one-to-one correspondence between solutions of the Euler-Lagrange equations for $\cL$ and the solutions of the Lagrangian-Hamiltonian problem of Definition \ref{def:k+1-unified}.
\end{thm}

A key feature of Definition \ref{def:k+1-unified} for the unified formalism is that no symmetry properties for the multimomenta are prescribed in advance. Rather, they have an arbitrary non symmetric part that does not play any role in the equations of motion.

\subsection{Hamiltonian formalism}
\label{sec:hamilt-form}

In the present section we will discuss how to relate the solutions of the Lagrangian-Hamiltonian problem with the solutions of the so called \emph{Hamilton equations of motion} (see Definition \ref{def:HamiltonEqs} below). Unlike with what happened in Lagrangian mechanics, things are not very straightforward in the Hamiltonian side, because up to now we do not have a projection from $\widetilde{W_\cL}$ onto a space of forms ($\widetilde{W_\cL}$ is a space of forms, but it contains too many velocities; we need a space of forms on $J^k\pi$ instead of on $J^{k+1}\pi$).

\subsubsection{A space of multimomenta}
\label{sec:space-multimomentum}

In order to define the space of multimomentum, let us define the map
\[
  R:\widetilde{W_\cL}\to J^{k+1}\pi\times_{J^k\pi}\wedge_2^m\left(J^k\pi\right)
\]
by the condition
\[
  R\left(\rho\right)=\left(j_x^{k+1}s,\rho'\right)
\]
if and only if
\[
  \tau_\cL\left(\rho\right)=j_x^{k+1}s\quad\text{ and }\quad\rho'\circ T_{j_x^{k+1}s}\pi_{k+1,k}=\rho.
\]
Locally, a $m$-form $\rho'$ in $\wedge^m_2\left(J^k\pi\right)$ can be written as
\[
  \rho'=q\eta+q^{I,i}du_{I}^\alpha\wedge\eta_i,\quad\abs{I}\leq k
\]
and so for $j_x^{k+1}s\in J^{k+1}\pi$ it results that
\[
  \rho'\circ T_{j_x^{k+1}s}\pi_{k+1,k}=q\eta+q_\alpha^{I,i}du_{I}^\alpha\wedge\eta_i=\left(q+q_\alpha^{\left(I,i\right)}u_{I+1_i}^\alpha\right)\eta+p^{I,i}_\alpha\theta^\alpha_{I}\wedge\eta_i;
\]
therefore the map $R$ is locally given by
\begin{equation}\label{eq:LocalMapR}
  q=L-p^{\left(I,i\right)}_\alpha u_{I+1_i}^\alpha,\qquad q^{I,i}_\alpha=p^{I,i}_\alpha.
\end{equation}
\begin{note}
  The map $R$ establishes an isomorphism of bundles on $J^k\pi$ between $\widetilde{W_\cL}$ and the subbundle $W_2$ defined in \cite{1751-8121-42-47-475207}. 
\end{note}
Then we can construct projections $p_\cL:\widetilde{W_\cL}\to\wedge^m_2\left(J^k\pi\right)$ and $p_\cL^\ddagger:=\mu\circ p_\cL$ making the following diagram commutative
\begin{equation}\label{eq:HamiltonDiagram}
  \begin{tikzcd}[row sep=1.8cm,column sep=2.3cm,ampersand replacement=\&]
    \widetilde{W_\cL}
    \arrow{dr}{p_\cL}
    \arrow[swap]{d}{\tau_\cL}
    \arrow{r}{R}
    \arrow[swap]{ddr}{p_\cL^\ddagger}
    \&
    J^{k+1}\pi\times_{J^k\pi}\wedge^m_2\left(J^k\pi\right)
    \arrow{d}{\text{pr}_2}
    \\
    J^{k+1}\pi
    \arrow{d}{\pi_{k+1,k}}
    \&
    \wedge^m_2\left(J^k\pi\right)
    \arrow{d}{\mu}
    \\
    {J^k\pi}
    \&
    J^{k+1}\pi^\ddagger
    \arrow{l}{\overline{\tau_k^m}}
  \end{tikzcd}    
\end{equation}
where $\mu:\wedge^m_2\left(J^k\pi\right)\to J^{k+1}\pi^\ddagger$ is the quotient map onto
\[
  J^{k+1}\pi^\ddagger:=\wedge^m_2\left(J^k\pi\right)/\wedge^m_1\left(J^k\pi\right),
\]
and we have introduced the handy notation
\[
  \overline{\tau^m_k}:J^{k+1}\pi^{\ddagger}\to J^k\pi.
\]
The map $p_\cL$ has nice properties regarding the canonical form on these bundles.
\begin{lem}\label{lem:hamilt-form}
  Let $\Theta\in\Omega^m\left(\wedge^m_2\left(J^k\pi\right)\right)$ be the canonical $m$-form on $\wedge^m_2\left(J^k\pi\right)$. Then
  \[
    p_\cL^*\Theta=\Theta_\cL.
  \]
\end{lem}
Now, consider a section
\[
  h:\wedge^m_2\left(J^k\pi\right)/\wedge^m_1\left(J^k\pi\right)\to\wedge^m_2\left(J^k\pi\right)
\]
for the quotient map $\mu:\wedge^m_2\left(J^k\pi\right)\to\wedge^m_2\left(J^k\pi\right)/\wedge^m_1\left(J^k\pi\right)$, that will be called \emph{Hamiltonian section associated to $\cL$}. Having a Hamiltonian section allows us to write down Hamilton equations of motion.
\begin{defc}[Hamilton equations for higher order field theories]\label{def:HamiltonEqs}
  Let $h$ be a Hamiltonian section for $\mu$; define
  \[
    \Theta_h:=h^*\Theta.
  \]
  A section $\psi:U\subset M\to J^{k+1}\pi^{\ddagger}$ is a \emph{solution for the Hamilton equations posed by the Hamiltonian section $h$} if and only if
  \[
    \psi^*\left(X\lrcorner d\Theta_h\right)=0
  \]
  for every $X\in\mathfrak{X}^{V\left(\pi_k\circ\overline{\tau^m_k}\right)}\left(J^{k+1}\pi^{\ddagger}\right)$.
\end{defc}

\subsubsection{The Hamiltonian form associated to a Lagrangian density}
\label{sec:hamilt-sect-assoc}

Given a Lagrangian density, the usual way to construct a Hamiltonian field theory is to construct a Hamiltonian section associated to this Lagrangian. The main drawback of this approach, when one have in mind the generalization of this procedure to general variational problems, is that no notion of extended multimomentum bundle can be found in the general setting. So it is necessary to find a way to avoid this difficulty; a close examination of the construction of Hamilton equations tells us that the only role of the Hamiltonian section is to provide a manner to define the Hamiltonian section. In this vein, it could be advantageous to find a procedure allowing us to construct the Hamiltonian form directly from the Lagrangian-Hamiltonian formalism, without recurring to the Hamiltonian section. 

\begin{defc}\label{def:FirstConstraintClass}
  The \emph{first constraint manifold} $P_0\subset\widetilde{W_\cL}$ is the set
  \[
    P_0:=\left\{\rho\in\widetilde{W_\cL}:\left(\cL_Z\Theta_\cL\right)\left(\rho\right)=0\text{ for all }Z\in\mathfrak{X}^{Vp_\cL^\ddagger}\left(\widetilde{W_\cL}\right)\right\}.
  \]
\end{defc}
\begin{note}[Local description for $P_0$]\label{rem:LocalDescriptionP_0}
    In every set of adapted coordinates $\left(x^i,u^\alpha,u^\alpha_I,p_\alpha^{J,i}\right)$ on $\widetilde{W_\cL}$, the equations
  \[
    p^{\left(I,i\right)}_\alpha=\frac{\partial L}{\partial u^\alpha_{I+1_i}},\qquad\abs{I}=k,
  \]
  are a local description for $P_0$. 
\end{note}
\begin{note}\label{rem:ThetaLLocal}
From the local expression for $\Theta_\cL$
\[
  \Theta_\cL=\left(L-p^{\left(I,i\right)}_\alpha u^\alpha_{I+1_i}\right)\eta+p^{I,i}_\alpha du^\alpha_I\wedge\eta_i
\]
it follows that on $P_0$, the local functions
\[
  \widehat{H}\left(x^i,u^\alpha_I,p_\beta^{I,i}\right)=L\left(x^i,u_I^\alpha\right)-p^{\left(I,i\right)}_\alpha u^\alpha_{I+1_i}
\]
are independent of the coordinates $u^\alpha_I,\abs{I}=k+1$.
\end{note}
Let $i_0:P_0\hookrightarrow\widetilde{W_\cL}$ be the canonical immersion, and define $p_0:=p_\cL^\ddagger\circ i_0$; also, define the set
\[
  C_0:=p_0\left(P_0\right)\subset J^{k+1}\pi^\ddagger.
\]
We will assume that this set is a submanifold of $J^{k+1}\pi^\ddagger$. From Remark \ref{rem:ThetaLLocal} it results that this function induces a local function $\widetilde{H}$ on $C_0$. A crucial fact about the set $P_0$ is that it can be described as the maximal set in which the form $\Theta_\cL$ is horizontal respect to the projection $p_\cL^\ddagger$; with this fact in mind, the following consequence of Lemma \ref{lem:hamilt-form} can be derived.
\begin{lem}\label{cor:Hamilt-form-2}
  Let $\cL:J^{k+1}\pi\to\wedge^m\left(T^*M\right)$ be a Lagrangian density. Then, there exists an $m$-form $\Theta'\in\Omega^m\left(C_0\right)$ such that
  \begin{equation}\label{eq:ThetaHThetaL}
    p_0^*\Theta'=i_0^*\Theta_\cL.
  \end{equation}
\end{lem}
We will see that this lemma characterizes uniquely the Hamiltonian form, without the use of a Hamiltonian section.
\begin{prop}[Hamiltonian form in classical variational problems]\label{Ejem:hamilt-form-from-map}
  The form $\Omega'$ is the Hamiltonian form associated to the Lagrangian density $\cL:J^{k+1}\pi\to\wedge^m\left(T^*M\right)$.
\end{prop}
\begin{proof}
  Let us see that this definition gives rise to the usual Hamiltonian form in the case of a classical variational problem; namely, we will show that $\Theta'$ from Lemma \ref{cor:Hamilt-form-2} becomes the usual Hamiltonian form
  \[
    \Theta_h=-\widetilde{H}\eta+p_\alpha^{\left(I,i\right)}du^\alpha_I\wedge\eta_i
  \]
  Now, from Remark \ref{rem:ThetaLLocal} and the local characterization for $P_0$ given in Remark \ref{rem:LocalDescriptionP_0}, we have
  \[
    \Theta_0=\widehat{H}\left(x^i,u^\alpha_I,p^{J,i}_\alpha,\frac{\partial L}{\partial u_{I+1_i}}\right)\eta+p^{J,i}_\alpha du^\alpha_J\wedge\eta_i+\left(p^{\left[I,i\right]}_\alpha+\frac{\partial L}{\partial u_{I+1_i}}\right)du^\alpha_I\wedge\eta_i,
  \]
  where we have employed the convention $\abs{J}<k,\abs{I}=1$ for the multiindices used in this expression. Because $p_\cL^\ddagger$ is surjective and $\Theta_0$ is $2$-vertical with respect to the projection $\pi_{k+1}\circ\tau_\cL$, the form $\Theta'$ should have the form
  \[
    \Theta'=M\eta+N^{I,i}_\alpha du^\alpha_I\wedge\eta_i+S^{\alpha,j}_{I,i}dp^{I,i}_\alpha\wedge\eta_j,
  \]
  and so it results 
  \[
    M=\widetilde{H},\qquad N^{I,i}_\alpha=p^{I,i}_\alpha,\qquad S^{\alpha,j}_{I,i}=0,
  \]
  giving rise to the Hamiltonian form $\Theta_h$ from Definition \ref{def:HamiltonEqs}. 
\end{proof}
Therefore, Definition \ref{def:GeneralHamForm} provides us with a correct generalization for this form, without having to appeal to a Hamiltonian section, which is hard to define in the general case.

\subsubsection{The projection of solutions of the Lagrangian-Hamiltonian problem}
\label{sec:proj-solut-lagr}

We are ready to prove the easier part of the relationship between the unified formalism and its Hamiltonian counterpart. It is analogous to the situation discussed in Section \ref{sec:lagrangian-formalism}, where you first prove that solutions project adequately, and then attack the (in general harder) problem of lift them.

\begin{prop}
  Let $\sigma:U\subset M\to P_0\subset\widetilde{W_\cL}$ be a solution for the Lagrangian-Hamiltonian problem posed by $\cL$ with associated Hamiltonian form $\Omega_h$. Then the section
  \[
    \psi:=p_0\circ\sigma:U\to C_0\subset J^{k+1}\pi^{\ddagger}
  \]
  is a solution for the Hamilton equations determined by the Hamiltonian form $\Omega_h$.
\end{prop}
\begin{proof}
  We have to prove that
  \begin{equation}\label{eq:HamiltonFieldTheory}
    \psi^*\left(X\lrcorner d\Theta_h\right)=0
  \end{equation}
  for every $X\in\mathfrak{X}^{V\left(\pi_k\circ\overline{\tau^m_k}\right)}\left(C_0\right)$, using the fact that $\sigma$ is a solution for the Lagrangian-Hamiltonian problem associated to $\cL$. Now, the map $p_0:P_0\to C_0$ is a surjective submersion, and so we can restrict the vector fields in \eqref{eq:HamiltonFieldTheory} to be of the form
  \[
    X:=Tp_0\circ Y
  \]
  for $Y\in\mathfrak{X}^{V\left(\pi_{k+1}\circ\tau_0\right)}\left(P_0\right)$, with $\tau_0:P_0\to J^{k+1}\pi$ the restriction of $\tau_\cL$ to $P_0\subset\widetilde{W}_\cL$. Therefore the Hamilton equations for $\psi$ will read
  \begin{align*}
    \psi^*\left(X\lrcorner d\Theta_h\right)&=\left(p_0\circ\sigma\right)^*\left(\left(Tp_0\circ Y\right)\lrcorner d\Theta_h\right)\\
                                           &=\sigma^*\left(Y\lrcorner d\left(p_0\right)^*\Theta_h\right)\\
    &=\sigma^*\left(Y\lrcorner d\Theta_0\right)
  \end{align*}
  where Lemma \ref{cor:Hamilt-form-2} was used. Then
  \[
    \psi^*\left(X\lrcorner d\Theta_h\right)=\sigma^*\left(Y\lrcorner d\Theta_0\right)=0,
    \]
    and so $\psi$ is a solution for the Hamilton equations, as required.
\end{proof}

\begin{note}[On singular Lagrangians]
  This result should be compared with Proposition $1$ from \cite{Prieto-Martinez2015203}. According to these authors, in the singular Lagrangian case, the construction of the solution of Hamilton equations require the projection onto the Lagrangian side and the translation along the Legendre transformation. The proof given here seems to work well without having regularity conditions imposed on the Lagrangian.
\end{note}

\section{The Hamilton equations for general variational problems}
\label{sec:hamilt-equat-gener-2}

We will devote this section to construct a set of Hamilton-like equations of motion for a general variational problem $\left(\pi:W\to M,\lambda,\cI\right)$. With the aim of carry out this task, we will use the general construction devised in \cite{GotayCartan} in order to formulate a unified formalism for a general variational problem. As we showed in the previous section for the classical variational problem case, the Hamilton equations can be obtained directly from this formulation, without the need of defining a Hamiltonian section.

\subsection{Unified formalism for general variational problems}
\label{sec:unif-form-gener}

Let us suppose that we have a general variational problem given by the triple
\begin{equation}\label{eq:GeneralVariationalProblem}
  \left(\pi:W\to M,\lambda,\cI\right),
\end{equation}
where $\pi:W\to M$ is a bundle, $\lambda\in\Omega^m\left(W\right)$ is the so called \emph{Lagrangian form} and $\cI$ is an ideal in $\Omega^\bullet\left(W\right)$ closed by exterior differentiation. The underlying variational problem is given by the action
\[
  S\left[\sigma\right]:=\int_K\sigma^*\lambda,\qquad K\subset M
\]
restricted to sections $\sigma:M\to W$ such that $\sigma^*\cI=0$ (these kind of sections will be called \emph{admissible sections}).

\begin{example}[Classical variational problem]
  When $W=J^{k+1}\pi$, $\lambda=\cL$ and $\cI$ is the differential ideal on $J^{k+1}\pi$ generated by the contact forms, this variational problem reduces to the classical variational problem considered in Section \ref{sec:class-lepage-equiv}.
\end{example}

\begin{example}[Vakonomic Herglotz variational problem]
  \label{ex:exampl-vakon-hergl}
  Let us consider the Herglotz variational problem for this viewpoint \cite{gaset2022herglotz}. For $\pi:P\to M$ let us define the \emph{Herglotz bundle}
  \[
    H_\pi:=J^1\pi\times_M\wedge^{m-1}\left(T^*M\right).
  \]
  In terms of local coordinates, any element $z\in\wedge^{m-1}\left(T^*M\right)$ can be written as
  \[
    z=z^i\eta_i.
  \]
  Note that we have a canonical $\left(m-1\right)$-form $\Theta_{m-1}\in\Omega^{m-1}\left(\wedge^{m-1}\left(T^*M\right)\right)$ that can be pulled back to $H_\pi$; we will indicate this pullback with the same symbol. Also, given a Lagrangian density $\cL:J^1\pi\to\wedge^m\left(T^*M\right)$ considered as a horizontal $m$-form on $J^1\pi$, we can also pull it back to the Herglotz bundle. Therefore, we can define a $m$-form on $H_\pi$
  \[
    \Phi:=\cL-d\Theta_{m-1}.
  \]
  \begin{defc}[Vakonomic Herglotz variational problem]
    The \emph{Vakonomic Herglotz variational problem} is the triple
    \[
      \left(H_\pi\to M,d\Theta_{m-1},\cJ_H\right)
    \]
    where $\cJ_H$ is the exterior differential system on $H_\pi$ generated by the (pullback of) contact forms on $J^1\pi$ and $\Phi$.
  \end{defc}
  In local coordinates, it means that $\lambda=dz^i\wedge\eta_i$ and every section $\sigma\left(x\right)=\left(x,f^\alpha\left(x\right),f^\alpha_i\left(x\right),\zeta^i\left(x\right)\right)$ will be admissible if and only if
  \[
    f^\alpha_i=\frac{\partial f^\alpha}{\partial x^i},\qquad\frac{\partial\zeta^i}{\partial x^i}=L.
  \]
  The action on any section will read
  \[
    S\left[\sigma\right]=\int_K\sigma^*d\Theta_{m-1}=\int_{\partial K}\sigma^*\Theta_{m-1}.
  \]
\end{example}
In order to apply the Gotay procedure \cite{GotayCartan} for the construction of an unified formalism for a variational problem of this kind, a regularity condition for the ideal $\cI$ is needed.
\begin{defc}[Regular variational problem]
  We will say that the variational problem \eqref{eq:GeneralVariationalProblem} is \emph{regular} if and only if the set of $m$-forms in $\cI$ that are $k$-vertical, namely
  \[
    \cI^m_k:=\cI\cap\Omega^m_k\left(W\right),
  \]
  can be generated by the sections of a bundle of $m$-forms $I^m_K\subset\Omega^m\left(W\right)$.
\end{defc}

Assuming that we are working with a regular variational problem, we can construct the affine subbundle
\[
  W_\lambda:=\lambda+I_k^m\subset\wedge^m\left(T^*W\right).
\]
Because it is a submanifold of a set of $m$-forms, it results that we can pull the canonical $m$-form on $\wedge^m\left(T^*W\right)$ back to $W_\lambda$; we will indicate this pullback with the symbol $\Theta_\lambda$.
\begin{defc}[Unified formalism for regular variational problems]\label{def:UnifiedFormalismGeneral}
  The \emph{unified formalism for general variational problems of regular nature} is the variational problem posed by the following data
  \[
    \left(W_\lambda\to M,\Theta_\lambda,0\right).
  \]
\end{defc}
\begin{note}
  It is an invariant procedure that allows us to incorporate the set of restrictions in $\cI$ as part of the Lagrangian through Lagrange multipliers. These multipliers play the role of (multi) momentum from this point of view. From this viewpoint, the \emph{regularity order} $k$ in the previous definition is a measure of the number of multimomentum necessary to allow the resulting unified formalism to be able to represent the Lagrangian equations of motion, and it is (more or less) described by the order of verticality of the forms generating $\cI$.
\end{note}

\begin{example}[Unified formalism for vakonomic Herglotz variational problem]\label{def:UnifiedHerglotz}
  Because we have found a formulation of the Herglotz variational problem as a general variational problem (see Example \ref{ex:exampl-vakon-hergl}), Definition \ref{def:UnifiedFormalismGeneral} will allow us to define a unified formalism for it. The set of restrictions $\cJ$ are regular for $k=2$, because the set of $m$-forms which are $2$-vertical in this ideal is generated by sections of the subbundle
  \[
    J^m_H:=\left\{p_\alpha^i\left(du^\alpha-u^\alpha_kdx^k\right)+\mu\left(dz^i\wedge\eta_i-L\eta\right):p_\alpha^i,\mu\in\mR\right\}.
  \]
  Therefore, we can define the subbundle
  \[
    W_H:=d\Theta_{m-1}+J^m_H\subset\wedge^m\left(T^*H_\pi\right)
  \]
  and pullback the canonical $m$-form on $\wedge^m\left(T^*H_\pi\right)$ to $W_H$, giving rise to the Lagrangian form $\Theta_H$ which in local coordinates reads
  \[
    \Theta_H=dz^i\wedge\eta_i+p_\alpha^i\left(du^\alpha-u^\alpha_kdx^k\right)+\mu\left(dz^i\wedge\eta_i-L\eta\right).
  \]
  \begin{defc}[Unified formalism for Herglotz variational problem]
    The \emph{unified formalism for the vakonomic Herglotz variational problem} is the variational problem posed by the data
    \[
      V_H':=\left(p_H:W_H\to M,\Theta_H,0\right)
    \]
  \end{defc}
  Please note that the form $\Theta_H$ found here coincides with the (multicontact) \emph{Lagrangian form} $\Theta_\cL$ constructed in \cite{leon23:_multic}. The equations describing the extremals of this variational problem are
  \[
    \sigma^*\left(Z\lrcorner d\Theta_H\right)=0,\qquad Z\in\mathfrak{X}^{Vp_H}\left(W_H\right);
  \]
  Then we can establish the following result (see also \cite{gaset2022herglotz}).
  \begin{thm}
    A section of $p_H:W_H\to M$ is an extremal for the variational problem $V_H'$ if and only if the following equations hold
    \begin{align*}
      &d\mu\wedge\eta_i+\mu\frac{\partial L}{\partial z^i}\eta=0,\qquad dz^i\wedge\eta_i-L\eta=0,\qquad du^\alpha\wedge\eta_i-u^\alpha_i\eta=0,\\
      &dp_\alpha^i\wedge\eta_i-\mu\frac{\partial L}{\partial u^\alpha}\eta=0,\qquad \left(p_\alpha^i-\mu\frac{\partial L}{\partial u^\alpha_i}\right)\eta=0.\\
    \end{align*}
  \end{thm}
\end{example}

\subsection{Hamilton equations for general variational problems}
\label{sec:hamilt-equat-gener}

The first obstruction when dealing with general variational problems is the Hamiltonian section, as it plays a central role in the construction of the form $\Omega_h$. It is evident that this problem is related with the fact that the set $W_\lambda$ is akin to the first constraint submanifold $W_0$ of the usual approach to the unified formalism \cite{Prieto-Martinez2015203}, and this submanifold projects onto the restricted multimomentum bundle $J^{k+1}\pi^\ddagger$; therefore, as we do not have a natural candidate in the generalized context for the extended multimomentum bundle, it becomes problematic to define a meaningful substitute for the Hamiltonian section. For this reason, we will suppose here that we have a map $p_\lambda:W_\lambda\to W^\ddagger$ analogous to the map $p_\cL^\ddagger$, and we will try to formulate the Hamilton equations using only this structure.

\subsubsection{Hamiltonian form from $p_\lambda$}
\label{sec:hamilt-form-from}

The main property of the Hamiltonian form regarding the equations of motion is Corollary \ref{cor:Hamilt-form-2}. In order to proceed, it is necessary to ensure that the equations defining $P_0$ exists. Thus, let us consider the following definition.

\begin{defc}[Compatibility between a form and a projection]\label{def:CompatibilityFormProjection}
  We will say that $p_\lambda:W\to W^\ddagger$ and $\Theta_\lambda$ are \emph{compatible} if and only if
  \[
    Y\lrcorner\left(\cL_Z\Theta_\lambda\right)=0,\qquad Z\lrcorner\Theta_\lambda=0,
  \]
  for all $Y\in\mathfrak{X}^{V\left(\pi\circ\tau_\lambda\right)}\left(W_\lambda\right)$ and $Z\in\mathfrak{X}^{Vp_\lambda}\left(W_\lambda\right)$.
\end{defc}

\begin{note}[The compatibility condition in local coordinates]
  Let us suppose that $W^\ddagger$ is fibred on $M$, and $\left(x^i,z^A,w^P\right)$ is a set of local coordinates adapted to the projections $p_\lambda:W\to W^\ddagger$ and $W^\ddagger\to M$, in such a way that
  \[
    p_\lambda\left(x^i,z^A,w^P\right)=\left(x^i,w^P\right),
  \]
  and so
  \[
    Vp_\lambda=\left<\frac{\partial}{\partial z^A}\right>,\qquad V\left(\pi\circ\tau_\lambda\right)=\left<\frac{\partial}{\partial z^A},\frac{\partial}{\partial w^P}\right>.
  \]
  Then we can write
  \begin{multline*}
    \Theta_\lambda=\sum_{p+r+s=m}F_{i_1,\cdots,i_p,A_1,\cdots,A_r,P_1,\cdots,P_s}\cdot\\
    \cdot dx^{i_1}\wedge\cdots\wedge dx^{i_p}\wedge dz^{A_1}\wedge\cdots\wedge dz^{A_r}\wedge dw^{P_1}\wedge\cdots\wedge dw^{P_s},
  \end{multline*}
  and the conditions
  \[
    Z\lrcorner\Theta_\lambda=0
  \]
  translate into
  \[
    \Theta_\lambda=\sum_{p+r=m}F^{i_1,\cdots,i_p}_{P_1,\cdots,P_r}\eta_{i_1,\cdots,i_p}\wedge dw^{P_1}\wedge\cdots\wedge dw^{P_s}
  \]
  and so
  \[
    \cL_{\partial/\partial z^B}\Theta_\lambda=\sum_{p+r=m}\frac{\partial F^{i_1,\cdots,i_p}_{P_1,\cdots,P_r}}{\partial z^B}\eta_{i_1,\cdots,i_p}\wedge dw^{P_1}\wedge\cdots\wedge dw^{P_s}
  \]
  therefore, the conditions $Y\lrcorner\left(\cL_Z\Theta_\lambda\right)=0$ become
  \begin{align*}
    \frac{\partial F^{i_1,\cdots,i_p}_{P_1,\cdots,P_r}}{\partial z^B}&=0,\qquad p<m,p+r=m.
  \end{align*}
  For example, using the local expressions found in Section \ref{sec:local-expressions} for the unified formalism associated to the classical variational problem for a Lagrangian density $\cL$, Eq. \eqref{eq:LocalThetaL} can be written as
  \[
    \Theta_\cL=\left({L}-p^{I,i}_\alpha{u}_{I+1_j}^\alpha\right)\eta+\sum_{i=1}^m\sum_{0\leq\abs{I}\leq k}p_\alpha^{I,i}du^\alpha_I\wedge\eta_i,
  \]
  and these conditions are met, proving that $p^\ddagger_\cL:\widetilde{W}_\cL\to J^{k+1}\pi^\ddagger$ and the canonical form $\Theta_\cL$ are in this case compatible. 
\end{note}

Intuitively, it amounts for the form $\Theta_\lambda$ to be $p_\lambda$-semibasic, and the forms $\cL_Z\Theta_\lambda,Z\in\mathfrak{X}^{Vp_\lambda}\left(W_\lambda\right)$ to be horizontal forms with respect to the projection onto the base manifold $M$. Therefore, we can introduce the following definition, analogous to Definition \ref{def:FirstConstraintClass}.

\begin{defc}\label{def:FirstConstraint}
  Suppose that $p_\lambda$ and $\Theta_\lambda$ are compatible. The \emph{first constraint manifold} $P_0\subset W_\lambda$ is the set
  \[
    P_0:=\left\{\rho\in W_\lambda:\left(\cL_Z\Theta_\lambda\right)\left(\rho\right)=0\text{ and }\left(Z\lrcorner\Theta_\lambda\right)\left(\rho\right)=0\text{ for all }Z\in\mathfrak{X}^{Vp_\lambda}\left(W_\lambda\right)\right\}.
  \]
\end{defc}

\begin{note}\label{rem:AnalogousDefinition}
  Please note that this definition is indeed analogous to Definition \eqref{def:FirstConstraintClass}. Also, it can be said that $P_0$ is the biggest set in $W_\lambda$ where $\Theta_\lambda$ is $p_\lambda$-basic.
\end{note}

Let us indicate with $i_0:P_0\hookrightarrow W_\lambda$ the canonical immersion; define the $m$-form
\[
  \Theta_0:=i_0^*\Theta_\lambda\in\Omega^m\left(P_0\right).
\]
Let us suppose further that the set $C_0:=p_\lambda\left(P_0\right)\subset W^\ddagger$ is a submanifold; let us indicate with $p_0:P_0\to C_0$ the restriction of the map $p_\lambda$ to $P_0$. Then we have the following result.
\begin{lem}\label{lem:HamiltonFormProp}
  There exists $\Theta_h\in\Omega^m\left(C_0\right)$ such that
  \[
    p_0^*\Theta_h=\Theta_0.
  \]
\end{lem}
\begin{proof}
  We need to prove that $\Theta_0$ is a $p_0$-basic form. According to Definition \ref{def:FirstConstraint}, we know already that
  \[
    Z\lrcorner\Theta_0=0\qquad\text{ for all }Z\in\mathfrak{X}^{Vp_0}\left(P_0\right);
  \]
  so we also have to prove that
  \[
    \cL_Z\Theta_0=0.
  \]
  But for $Z\in\mathfrak{X}^{Vp_0}\left(P_0\right)$ we have that the image of
  \[
    Z_0:=Ti_0\circ Z
  \]
  belongs to $Vp_\lambda$, because
  \[
    Tp_\lambda\circ Z_0=Tp_\lambda\circ Ti_0\circ Z=Tp_0\circ Z=0.
  \]
  Therefore, $\cL_{Z_0}\Theta_\lambda=0$ and so
  \[
    0=i_0^*\left(\cL_{Z_0}\Theta_\lambda\right)=\cL_Z\Theta_0
  \]
  as required.
\end{proof}

\begin{defc}[Hamiltonian form]\label{def:GeneralHamForm}
  The form $\Theta_h\in\Omega^m\left(C_0\right)$ will be called \emph{Hamiltonian form} associated to the variational problem $\left(\pi:W\to M,\lambda,\cI\right)$ and the projection $p_\lambda$. 
\end{defc}

\begin{note}
  In line with Remark \ref{rem:AnalogousDefinition}, the Hamiltonian form is nothing but the form on $C_0$ whose pullback along $p_0$ gives rise to $\Theta_0$.
\end{note}

\subsubsection{Hamilton equations for general variational problems}
\label{sec:hamilt-equat-gener-1}

Using the notation of the previous section, Suppose that $q:W^\ddagger\to M$ gives $W^\ddagger$ a fiber bundle structure on $M$ such that the next diagram commutes
\[
  \begin{tikzcd}[row sep=1.8cm,column sep=2.3cm,ampersand replacement=\&]
    W_\lambda
    \arrow{r}{p_\lambda}
    \arrow[swap]{dr}{\pi\circ\tau_\lambda}
    \&
    W^\ddagger
    \arrow[swap]{d}{q}
    \\
    \&
    M
  \end{tikzcd}        
\]
With the help of the Hamiltonian form $\Omega_h:=d\Theta_h$ we can set the \emph{Hamilton equations}\footnote{The Hamilton-Cartan equations for classical variational problems have the additional requirement that its sections should be holonomic; this condition cannot be reproduced in the general case, because we do not have a double fibration structure such as $J^{k+1}\pi\to J^k\pi\to M$ in the latter.}.
\begin{defc}[Hamilton equations for a general variational problem]\label{def:hamilt-equat-gener}
  Let $i:C\hookrightarrow W^\ddagger$ a subbundle. The \emph{Hamilton equations on $C$ for a (regular) general variational problem} $\left(W\to M,\lambda,\cI\right)$ is the set of equations given by
  \begin{equation}\label{eq:GeneralHamiltonEqs}
    \psi^*\left(X\lrcorner\left(i^*\Omega_h\right)\right)=0\qquad\forall X\in\mathfrak{X}\left(C\right).
  \end{equation}
\end{defc}
We will indicate with $H_C\subset\Gamma\left(\left.q\right|_C\right)$ the set of (local) solutions for these equations. Accordingly, $C'\subset C$ will be the set on which the solutions of the Hamilton equations should live; we will call \emph{solution set} to this set. We will say that $C$ is \emph{final} if for every $\rho\in C$ there exists an open set $U\subset M$ and a solution $\psi:U\to C$ for the Hamilton equations such that $\rho\in\mathop{\text{Im}}{\psi}$, namely, if $C=C'$.

\subsubsection{Hamiltonian field theory for (vakonomic) Herglotz variational problem}
\label{sec:vakon-hergl-vari}

Recall from Example \ref{def:UnifiedHerglotz} that the variational problem
\[
  V'_H:=\left(W_H\to M,\Theta_H,0\right)
\]
provide us with a unified formalism for Herglotz variational problem. The construction just described will allow us to formulate a Hamiltonina field theory for this variatinal problem. So, let us define
\[
  W^\ddagger:=\wedge^m_2\left(T^*P\right)\times_M\wedge^{m-1}\left(T^*M\right).
\]
Then we have the following projection map
  \[
    p_H^\ddagger:W_H\to W^\ddagger:\rho_{\left(j_x^1s,\sigma\right)}\mapsto\left(\alpha,\beta,\nu\right)
  \]
  if and only if $\beta=\sigma$ and
  \[
    \rho=\alpha\circ T_{\left(\alpha,\beta,\nu\right)}\left(\text{pr}_1\circ\pi_{10}\right)+\left(1-\nu\right)\left.d\Theta_{m-1}\right|_\beta\circ T_{\left(\alpha,\beta,\nu\right)}\text{pr}_2.
  \]
  In coordinates, if $\alpha=q\eta+q_\alpha^idu^\alpha\wedge\eta_i$ and $\beta=w^i\eta_i$, this map is defined by the equations
  \[
    q=\mu L-p_\alpha^iu^\alpha_i,\qquad q_\alpha^i=p_\alpha^i,\qquad\nu=\mu.
  \]
  We have that
    \[
      Vp_H^\ddagger=\left<\frac{\partial}{\partial u_i^\alpha}\right>
    \]
    and is can be seen that $p_H^\ddagger$ and $\Theta_H$ are compatible in the sense of Definition \ref{def:CompatibilityFormProjection}, and so it is possible to construct a Hamiltonian field theory from this unified formulation.

    So the first constraint submanifold is this case will become
    \[
      P_0=\left\{\left(x^i,u^\alpha,u_i^\alpha,p_\alpha^i,z^i,\mu\right)\in W^\dagger:p_\alpha^i=\mu\frac{\partial L}{\partial u_i^\alpha}\right\};
    \]
    the phase space for the Hamiltonian field theory can be obtained from projection of this set along $p_H^\ddagger$, namely 
    \begin{align*}
      C_0&:=p_H^\ddagger\left(P_0\right)\cr
      &=\left\{\left(x^i,u^\alpha,p_\alpha^i,z^i,\mu\right)\in W^\dagger:p_\alpha^i=\mu\frac{\partial L}{\partial u_i^\alpha}\right\}.
    \end{align*}
    In order to construct the Hamiltonian form, we must pull the $m$-form $\Theta_H$ back to $P_0$; it turns out that
    \[
      i_0^*\Theta_H=-H_0\eta+p_\alpha^idu^\alpha\wedge\eta_i+\left(1-\mu\right)dz^i\wedge\eta_i,
    \]
    where $H_0=p_\alpha^iu_i^\alpha-\mu L$ is the energy function on $P_0$. Therefore the Hamiltonian form will result
    \[
      \Theta_h=-\widetilde{H}_0\eta+p_\alpha^idu^\alpha\wedge\eta_i+\left(1-\mu\right)dz^i\wedge\eta_i
    \]
    where $\widetilde{H}_0$ is the function on $C_0$ induced by $H_0$.
    \begin{thm}[Hamiltonian field theory for Herglotz variational problem]
      The \emph{Hamiltonian field theory for the Herglotz variatinal problem} is the pair $\left(C_0,\Omega_h\right)$. In this setting the Hamilton equations become
    \begin{align*}
      &dp_\alpha^i\wedge\eta_i+\frac{\partial\widetilde{H}_1}{\partial u^\alpha}\eta=0,\quad du^\alpha\wedge\eta_i-\frac{\partial\widetilde{H}_1}{\partial p_\alpha^i}\eta=0,\\ &d\mu\wedge\eta_i-\frac{\partial\widetilde{H}_1}{\partial z^i}\eta=0,\quad dz-Ldt=0.
    \end{align*}
  \end{thm}


\section{Hamilton equations for general relativity with basis}
\label{sec:introduction}

In this section a novel multisymplectic Hamiltonian scheme for gravity with basis (in the full frame bundle, whose structure group is $GL\left(n\right)$) is developed using the method described above. The idea is to extract this formalism using the description for GR with basis from \cite{doi:10.1142/S0219887818500445}; we will borrow the notation from this article. This construction can be compared with the formulation devised in \cite{0264-9381-32-9-095005}, which requires the structure group of the underlying principal bundle to be the Lorentz group, and uses in its design a classical variational problem. Another construction of a Hamiltonian theory for a first order version of gravity, in this case metric-affine gravity, is presented in \cite{gaset19:_new_metric_affin_einst_palat}; it should be noted that some extra constraints found in this article are absent in our scheme, because the general nature of the variational problem used in building the Hamiltonian theory below allows us to avoid the appearance of these constraints. Another important characteristic of this approach is that no Legendre transformation is needed in order to build the Hamilton equations; on the downside, it is necessary for the construction to be able to find a projection compatible with the multisymplectic form of the unified problem.

\subsection{Introduction and basic notation}
\label{sec:intr-basic-notat}

Let
\[
  \tau:LM\to M
\]
be the frame bundle on the manifold $M$. Therefore, we can define an unified problem on
\[
  \widehat{W_\cL}=J^1\tau\times_{M} E_2
\]
where
\[
  E_2:=\wedge^{m-1}M\otimes S^*\left(m\right),\qquad S^*\left(m\right):=\left(\mR^m\right)^*\odot\left(\mR^m\right)^*.
\]
The Lagrangian form on this space is given by
\[
  \lambda_\cL:=\eta^{kl}\theta_{kp}\wedge\Omega^p_l+\eta^{ql}\Theta_{pq}\wedge\omega^p_l,
\]
where $\Theta$ is the $S^*\left(m\right)$-valued $\left(m-1\right)$-form on $E_2$, $\omega$ is the canonical connection form on $J^1\tau$ and $\Omega$ its curvature form; $\theta_{pqr\cdots}$ indicate the Spalding forms of different orders, and $\eta$ is a matrix defining a Lorentz group in $GL\left(m\right)$. Given the affine structure of the jet bundle, the vertical spaces $V_{j_x^1s}\tau_{10}\subset T_{j_x^1s}\left(J^1\tau\right)$ are spanned by the vectors
\[
  \left(\theta^r,\left(E^p_q\right)_{LM}\right)^V\left(j_x^1s\right):=\left.\frac{\text{d}}{\text{d}t}\right|_{t=0}\left[j_x^1s+t\left.\theta^r\right|_{s\left(x\right)}\otimes\left(E^p_q\right)_{LM}\left(s\left(x\right)\right)\right],
\]
and the map
\[
  j_x^1s\mapsto\left(\theta^r,\left(E^p_q\right)_{LM}\right)^V\left(j_x^1s\right)
\]
defines a $\tau_{10}$-vertical vector field on $J^1\tau$ (and also on $\widehat{W_\cL}$). Recall that in the coordinates $\left(x^i,e_j^\mu\right)$, the infinitesimal generators for the $GL\left(m\right)$-action become
\[
  \left(E_p^q\right)_{LM}\left(x^i,e_j^\mu\right):=-e_p^\sigma\frac{\partial}{\partial e^\sigma_q};
\]
it means that
\begin{align}
  j_x^1s+t\left.\theta^r\right|_{s\left(x\right)}\otimes\left(E^p_q\right)_{LM}\left(s\left(x\right)\right)&=\left[\frac{\partial}{\partial x^\mu}\mapsto\frac{\partial}{\partial x^\mu}+e^\nu_{k\mu}\frac{\partial}{\partial e^\nu_k}-te^r_\mu e^\nu_q\frac{\partial}{\partial e^\nu_{p}}\right]\cr
  &=\left[\frac{\partial}{\partial x^\mu}\mapsto\frac{\partial}{\partial x^\mu}+\left(e^\nu_{k\mu}-te^r_\mu e_q^\nu\delta^p_k\right)\frac{\partial}{\partial e^\nu_k}\right]\label{eq:AffineActionJet}
\end{align}
for
\[
  j_x^1s:\frac{\partial}{\partial x^\mu}\mapsto\frac{\partial}{\partial x^\mu}+e^\nu_{k\mu}\frac{\partial}{\partial e^\nu_k}.
\]
Therefore, in the canonical coordinates $\left(x^i,e^\mu_k,e^\mu_{k\nu}\right)$ these vertical vectors are given by the formula
\[
  \left(\theta^r,\left(E^p_q\right)_{LM}\right)^V\left(j_x^1s\right)=-e^r_\nu e^\mu_q\frac{\partial}{\partial e^\mu_{p\nu}}.
\]
\begin{lem}
  There exists an action of the abelian group (under addition) $A:=\mR^m\otimes\left(\mR^m\right)^*\otimes\left(\mR^m\right)^*$ on $J^1\tau$ given by the formula
  \[
    c\cdot j_x^1s:=j_x^1s+c^i_{jk}\left.\theta^j\right|_{s\left(x\right)}\otimes\left(E^k_i\right)_{LM}\left(s\left(x\right)\right)
  \]
  where $c=\left(c_{ij}^k\right)$.
\end{lem}
\begin{note}[The action in coordinates]
  Let us see how this action operates at coordinate level. By Equation \eqref{eq:AffineActionJet}, we will obtain
  \[
    c\cdot e^\nu_{k\mu}=e^\nu_{k\mu}-c_{rk}^qe^r_\mu e^\nu_q.
  \]
\end{note}
Please note that the Lie algebra $\mathfrak{a}$ coincides with the space $A$; the infinitesimal generators for this action are the vertical vector fields
\[
  c_{J^1\tau}:=c_{rp}^q\left(\theta^r,\left(E^p_q\right)_{LM}\right)^V,\qquad c\in\mathfrak{a}.
\]
Now, we are searching for a projection $p^\ddagger:\widehat{W_\cL}\to W^\ddagger$ compatible with the form $\lambda_\cL$; the idea is to construct this projection by quotienting out by a subspace $B\subset A$. In order to proceed, we need the following lemma.

\begin{lem}
  It is true that
  \[
    \left(\theta^r,\left(E^p_q\right)_{LM}\right)^V\lrcorner\omega=0=\left(\theta^r,\left(E^p_q\right)_{LM}\right)^V\lrcorner\theta=\left(\theta^r,\left(E^p_q\right)_{LM}\right)^V\lrcorner T.
  \]
  Also
  \[
    \left(\theta^r,\left(E^p_q\right)_{LM}\right)^V\lrcorner\Omega^i_j=-\delta^p_j\delta^i_q\theta^r.
  \]
\end{lem}
It is important to have at our disposal the following expression for the differential of the Lagrangian form.
\begin{prop}
  The differential of the Lagrangian form for GR is given by
  \begin{multline}\label{eq:FormulaFordLambda0}
    \left.d\lambda_\cL\right|_\rho=\left[2\eta^{kp}\left(\omega_\pf\right)_k^i\wedge\theta_{il}-\left(\omega_\pf\right)^s_s\wedge\eta^{kp}\theta_{kl}+\eta^{kp}T^i\wedge\theta_{kli}+\eta^{ip}\left.\Theta_{il}\right|_{\beta}\right]\wedge\Omega^l_p+\\
    +\eta^{ik}\left[\left.d\Theta_{ij}\right|_{\beta}+\eta^{rq}\eta_{li}\left.\Theta_{rj}\right|_{\beta}\wedge\left(\omega_\kf\right)^l_q-\left.\Theta_{ip}\right|_{\beta}\wedge\left(\omega_\kf\right)^p_j\right]\wedge\left(\omega_\pf\right)^j_k,
  \end{multline}
  where $\omega_\kf$ (resp. $\omega_\pf$) are $\eta$-antisymmetric (resp. $\eta$-symmetric) components of $\omega$.
\end{prop}

\subsection{Compatibility in the affine action}
\label{sec:comp-affine-acti}

We will use the affine action defined above in order to find a projection compatible with $\lambda_\cL$ and thus Hamiltonian equations. To this end, let us explore the conditions for compatibility for the form $\lambda_{\cL}$. We have the formulas
\begin{align}
  0=c_{J^1\tau}\lrcorner\lambda_{\cL}&=\left(-1\right)^{n+1}\eta^{kl}c_{rt}^s\delta^p_s\delta^t_l\theta_{kp}\wedge\theta^r\cr
                                   &=\left(-1\right)^{n+1}\eta^{kt}c_{rt}^s\theta_{ks}\wedge\theta^r\cr
                                   &=\eta^{kt}c_{rt}^s\left(\delta^r_k\theta_s-\delta^r_s\theta_s\right)\cr
                                   &=\left(\eta^{kt}c_{kt}^s-\eta^{st}c_{rt}^r\right)\theta_s\label{eq:CContractedLambda}
\end{align}
and
\begin{align}
  &c_{J^1\tau}\lrcorner d\lambda_{\cL}=\cr
  &=\left(-1\right)^{n+1}\left[2\eta^{kp}\left(\omega_\pf\right)_k^i\wedge\theta_{il}-\left(\omega_\pf\right)^s_s\wedge\eta^{kp}\theta_{kl}+\eta^{kp}T^i\wedge\theta_{kli}+\eta^{ip}\left.\Theta_{il}\right|_{\beta}\right]\wedge c_{rs}^t\delta^r_p\delta^l_t\theta^s.\label{eq:dLContractedC}
\end{align}
The last two terms give rise to horizontal forms, so they are admissible; on the other hand, the remaining terms should annihilate for compatibility. Then the element $c\in\mathfrak{a}=A$ must be chosen such that
\begin{align*}
  0&=\eta^{ks}c_{rs}^l\eta^{ks}\left[2\left(\omega_\pf\right)^i_k\wedge\theta_{is}\wedge\theta^r-\left(\omega_\pf\right)^i_i\wedge\theta_{kl}\wedge\theta^r\right]\\
   &=2\left(\eta^{ks}c_{rs}^r\theta_i-\eta^{ks}c_{is}^l\theta_l\right)\wedge\left(\omega_\pf\right)^i_k-\left(\eta^{ks}c_{rs}^r\theta_k-\eta^{ks}c_{ks}^l\theta_l\right)\wedge\left(\omega_\pf\right)^i_i.
\end{align*}
The second term in this equation is equivalent to Equation \eqref{eq:CContractedLambda}; therefore, $c$ should be selected in order to ensure the following equation holds
\[
  \left(\eta^{ks}c_{rs}^r\theta_i-\eta^{ks}c_{is}^l\theta_l\right)\wedge\left(\omega_\pf\right)^i_k=0.
\]
When contracting with an infinitesimal generator for the $GL\left(m\right)$-action on $J^1\tau$, the form $\omega_\pf$ gives an element in the subspace $\pf$, which means that the term between parenthesis must live in $\kf$. Therefore, this condition translates into 
\begin{equation}\label{eq:CompatibilityConds}
  c_{pq}^l+c_{qp}^l-c^r_{rq}\delta_p^l-c^r_{rp}\delta_q^l=0.
\end{equation}
Contracting it with $\eta^{pq}$, it reduces to Equation \eqref{eq:CContractedLambda}, so that this condition is sufficient to define the subspace $B\subset A$.
\begin{lem}\label{lem:SolCompatibility}
  The solutions of \eqref{eq:CompatibilityConds} are elements $c\in A$ of the form
  \[
    c_{pq}^r=b_p\delta_q^r+b_q\delta_p^r+a_{pq}^r
  \]
  where $b_p=\left(1/2\right)c_{sp}^s$ and the components $a_{pq}^r$ satisfy the conditions
  \[
    a_{pq}^r+a_{qp}^r=0,\qquad a_{rq}^r=\frac{1-n}{2}c_{rq}^r.
  \]
\end{lem}
\begin{proof}
  We split $c$ into symmetric and antisymmetric parts
  \[
    c_{pq}^r=s_{pq}^r+a_{pq}^r,\qquad s_{pq}^r=\frac{1}{2}\left(c_{pq}^r+c_{qp}^r\right),\quad a_{pq}^r=\frac{1}{2}\left(c_{pq}^r-c_{qp}^r\right).
  \]
  According to Equation \eqref{eq:CompatibilityConds}, we have that
  \[
    s_{pq}^r=\frac{1}{2}\left(c_{sp}^s\delta_q^r+c_{sq}^s\delta_p^r\right);
  \]
  therefore, from
  \[
    c_{tq}^t=s_{tq}^t+a_{tq}^t
  \]
  we obtain the equation
  \[
    c_{tq}^t=\frac{n+1}{2}c_{tq}^t+a_{tq}^t
  \]
  for the trace of the antisymmetric part.
\end{proof}
\begin{cor}\label{cor:SplittingSolution}
  Any element $e\in A$ can be written as
  \[
    e_{pq}^r=\frac{1}{2}\left(e_{pq}^r+e_{qp}^r\right)-\frac{1}{2}\left(e_{sp}^s\delta_q^r+e_{sq}^s\delta_p^r\right)+a_q\delta_p^r-a_p\delta_q^r+c_{pq}^r
  \]
  where $c$ is a solution of Equation \eqref{eq:CompatibilityConds} and $a_q$ is chosen to adjust the trace of the antisymmetric part of $c$, namely
  \[
    a_q=\frac{1}{2}e_{sq}^s+\frac{1}{2\left(n-1\right)}\left(e_{sq}^s-e_{qs}^s\right)=\frac{n}{2\left(n-1\right)}e_{sq}^s-\frac{1}{2\left(n-1\right)}e_{qs}^s.
  \]
\end{cor}

\subsection{The quotient bundle and its primary constraint submanifold}
\label{sec:quotient-bundle-its}

It is now necessary to define the bundle where the Hamilton equations should live; it will be done through quotient. After that, the primary constraint submanifold could be defined and so the Hamiltonian form.

\subsubsection{A projection compatible with $\lambda_\cL$}
\label{sec:proj-comp-with}

We will define
\[
  W^\ddagger:=\widehat{W_\cL}/B
\]
with canonical projection $p^\ddagger:=p^{\widehat{W}_\cL}_B$. We can use Corollary \ref{cor:SplittingSolution} in order to define coordinates on this quotient bundle. Using the definition
\[
  -e^\mu_qe^\nu_pe_\sigma^rM_{\mu\nu}^\sigma=\frac{1}{2}\left(e_{pq}^r+e_{qp}^r\right)-\frac{1}{2}\left(e_{sp}^s\delta_q^r+e_{sq}^s\delta_p^r\right)+a_q\delta_p^r-a_p\delta_q^r
\]
for the part in an element in $\mathfrak{a}$ complementary to $\mathfrak{b}$, and recalling that
\[
  e_{pq}^r=-e_p^\nu e_\mu^re_q^\sigma\Gamma_{\sigma\nu}^\mu,
\]
we obtain
\[
  M_{\mu\nu}^\sigma=\frac{1}{2}\left(\Gamma_{\mu\nu}^\sigma+\Gamma_{\nu\mu}^\sigma\right)+\frac{1}{2\left(n-1\right)}\left(\Gamma_{\rho\nu}^\rho-\Gamma_{\nu\rho}^\rho\right)\delta_\mu^\sigma+\frac{1}{2\left(n-1\right)}\left(\Gamma_{\rho\mu}^\rho-\Gamma_{\mu\rho}^\rho\right)\delta_\nu^\sigma.
\]
Therefore we can consider $\left(x^\mu,e^\mu_k,M_{\mu\nu}^\sigma,\beta_{pq}^r\right)$ as a set of coordinates on the quotient bundle $W^\ddagger$; we need to know the primary constraint submanifold in order to construct the Hamiltonian form associated to this projection.

\subsubsection{On the geometrical interpretation of the phase space for Hamilton equations}
\label{sec:geom-interpr-hamilt}

In the discussion we will carry out in the next sections, it will be assumed that the $M$-variables in the phase space $W^\ddagger$ correspond to connection symbols for some affine connection on the spacetime $M$. It seems necessary to explain why it is possible to do so; after all, the phase space was created through a quotient procedure, and any individual with experience in this kind of operations knows that it tends to obliterate any structure that the original space possesses. The key is to interpret geometrically the action of the group $B$; in order to proceed in this direction, the following concept could be helpful \cite{MR532831}.

\begin{defc}[Projectively equivalent connections]
  Connections $\nabla,\nabla'$ on a manifold $P$ are \emph{projectively equivalent} if and only if they share the same unparametrized geodesics.
\end{defc}
In this vein, the following result is useful in characterizing projectively equivalent connections.
\begin{prop}
  $\nabla,\nabla'$ connections on $P$ are projectively equivalent if and only if there exists a $1$-form $\omega$ on $P$ such that
  \[
    \nabla'_XY=\nabla_XY+\omega\left(X\right)Y+\omega\left(Y\right)X
  \]
  for every $X,Y\in\mathfrak{X}\left(P\right)$.
\end{prop}
At Christoffel symbols level it means that
\[
  \left(\Gamma'\right)_{ij}^k=\Gamma_{ij}^k+b_i\delta_j^k+b_j\delta_i^k,
\]
where $b_i$ are the components of the $1$-form $\omega$ in the chosen coordinates. So, essentially, each class in $W^\ddagger$ represents a connection up to a torsion and a reparametrization of its geodesics. But it should be kept in mind that the trace of the antisymmetric part of any element in $B$ is fixed by the form associated to the reparametrization part. It is this fact that fixes the right parametrization for a geodesic once the torsion vanishes (it is ultimately obtained from the Hamilton equations, see Theorem \ref{thm:HamiltonEqsGR}).

\subsubsection{The primary constraint submanifold}
\label{sec:prim-constr-subm}

As we know, the vector subbundle $Vp^\ddagger$ is spanned by elements of the form $c_{J^1\tau}$, with $c\in B$. From Equation \eqref{eq:CContractedLambda} we have that there are no constraints associated to the equation
\[
  Z\lrcorner\lambda_\cL=0.
\]
From Equation \eqref{eq:dLContractedC} the other constraints become
\begin{align*}
  0=&\left(\eta^{kp}T^i\wedge\theta_{kli}+\eta^{ip}\left.\Theta_{il}\right|_{\beta}\right)\wedge c_{rs}^t\delta^r_p\delta^l_t\theta^s\\
  &=c_{rs}^t\left(\eta^{kr}T^i\wedge\theta_{kti}+\eta^{ir}\beta_{it}^q\theta_q\right)\wedge\theta^s\\
  &=\left(-1\right)^{n+1}c_{rs}^t\left[\eta^{kr}T^i_{pq}\theta^p\wedge\theta^q\wedge\left(\theta_{kt}\delta_i^s-\theta_{ki}\delta_t^s+\theta_{ti}\delta_k^s\right)+\eta^{ir}\beta_{it}^q\delta_q^s\sigma_0\right]
\end{align*}
where we have written
\[
 T^i:=T^i_{pq}\theta^p\wedge\theta^q.
\]
But now
\begin{align*}
  \theta^p\wedge\theta^q\wedge&\left(\theta_{kt}\delta_i^s-\theta_{ki}\delta_t^s+\theta_{ti}\delta_k^s\right)=\\
                              &=\theta^p\wedge\left[\left(\delta_t^q\theta_{k}-\delta_k^q\theta_t\right)\delta_i^s-\left(\delta_i^q\theta_{k}-\delta_k^q\theta_i\right)\delta_t^s+\left(\delta_i^q\theta_{t}-\delta_t^q\theta_i\right)\delta_k^s\right]\\  &=\left(\delta_t^q\delta^p_{k}-\delta_k^q\delta_t^p\right)\delta_i^s-\left(\delta_i^q\delta^p_{k}-\delta_k^q\delta^p_i\right)\delta_t^s+\left(\delta_i^q\delta^p_{t}-\delta_t^q\delta^p_i\right)\delta_k^s.
\end{align*}
Replacing in the previous equation and taking into account the antisymmetry of $T_{pq}^i$, we obtain
\begin{align*}
  0&=c_{rs}^t\left[2\eta^{kr}T^i_{pq}\left(\delta_t^q\delta_k^p\delta_i^s-\delta_i^q\delta_k^p\delta_t^s+\delta_i^q\delta_t^p\delta_k^s\right)+\eta^{ir}\beta_{it}^q\right]\\
  &=\eta^{kr}\left(2c_{rs}^tT_{kt}^s-2c_{rt}^tT_{ki}^i+2c_{rk}^tT_{ti}^i+c_{rs}^t\beta_{kt}^s\right)\\
   &=2\eta^{kr}\left(-\eta^{kr}c_{rt}^t+\eta^{rl}c_{rl}^k\right)T_{ki}^i+\eta^{kr}\left(2c_{rs}^tT_{kt}^s+c_{rs}^t\beta_{kt}^s\right)\\
  &=\eta^{kr}c_{rs}^t\left(2T_{kt}^s+\beta_{kt}^s\right)
\end{align*}
where the first term annihilates because $c\in B$.

\begin{prop}\label{prop:PrimaryConstraint}
  An element $\left(x^\mu,e^\mu_k,e_{k\nu}^\sigma,\beta_{pq}^r\right)$ belongs to the primary constraint submanifold if and only if
  \[
    \eta^{kr}c_{rs}^t\left(2T_{kt}^s+\beta_{kt}^s\right)=0
  \]
  for every $c\in B$.
\end{prop}
\begin{note}
  We will indicate the primary constraint submanifold with the symbol
  \[
    W_0:=\left\{\left(x^\mu,e^\mu_k,e_{k\nu}^\sigma,\beta_{pq}^r\right)\in\widehat{W_\cL}:\eta^{kr}c_{rs}^t\left(2T_{kt}^s+\beta_{kt}^s\right)=0\right\},
  \]
  and we will use $p_0:W_0\to W^\ddagger$ for the restriction of the projection $p^\ddagger$ to this submanifold.
\end{note}
There are some important properties of the elements of the primary constraint submanifold.
\begin{prop}
  If $\left(x^\mu,e^\mu_k,M_{\mu\nu}^\sigma,\beta_{pq}^r\right)$ is an element of the primary constraint submanifold, then
  \[
    \eta^{kr}\beta_{kr}^s=0=T_{kt}^t=\beta_{kt}^t.
  \]
\end{prop}
\begin{proof}
  We will try with different elements in $B$ in order to obtain constraints on $\beta$ and $T$. For example, consider
  \[
    c^t_{rs}=b_r\delta_s^t+b_s\delta_r^t+a_r\delta_s^t-a_s\delta_r^t
  \]
  where $b_r=(1/2)c^s_{rs}$ and $a_r$ is chosen to ensure that $c$ has the right trace properties, namely
  \begin{align*}
    c_{rs}^s&=\left(n+1\right)b_r+\left(n-1\right)a_r\\
            &=\frac{n+1}{2}c_{rs}^s+\left(n-1\right)a_r,
  \end{align*}
  implying that
  \[
    a_r=-\left(1/2\right)c_{rs}^s.
  \]
  Using it in Proposition \ref{prop:PrimaryConstraint} gives us
  \begin{align*}
    0&=\eta^{kr}\left(b_r\delta_s^t+b_s\delta_r^t+a_r\delta_s^t-a_s\delta_r^t\right)\left(2T_{kt}^s+\beta_{kt}^s\right)\\
     &=\eta^{kr}b_r\left(2T_{ks}^s+\beta_{ks}^s\right)+b_s\eta^{kr}\left(2T_{kr}^s+\beta_{kr}^s\right)+\eta^{kr}a_r\left(2T_{ks}^s+\beta_{ks}^s\right)-a_s\eta^{kr}\left(2T_{kr}^s+\beta_{kr}^s\right)\\
     &=\eta^{kr}\left(2T_{ks}^s+\beta_{ks}^s\right)\left(a_r+b_r\right)+\eta^{kr}\beta_{kr}^s\left(b_s-a_s\right),
  \end{align*}
  namely,
  \begin{equation}\label{eq:FirstEqFirstConstraintManifold}
    \eta^{kr}\beta_{kr}^s=0.
  \end{equation}
  Also, any element of the form $a_{pq}^r$ with the properties
  \[
    a_{pq}^r+a_{qp}^r,\qquad a_{pr}^r=0
  \]
  belongs to $B$; as a particular element with these properties, we can use
  \[
    a_{rs}^t=m_{rs}b^t+a_r\delta_s^t-a_s\delta_r^t
  \]
  with $m_{pq}+m_{qp}=0$ and $a_s$ is chosen to ensure this element has zero trace, so that
  \[
    0=m_{rs}b^s+\left(n-1\right)a_r.
  \]
  Then
  \begin{align*}
    0&=\eta^{kr}\left(m_{rs}b^t+a_r\delta_s^t-a_s\delta_r^t\right)\left(2T_{kt}^s+\beta_{kt}^s\right)\\
     &=\eta^{kr}\left[m_{rs}\left(2T_{kt}^s+\beta_{kt}^s\right)-\frac{1}{n-1}m_{rt}\left(2T_{ks}^s+\beta_{ks}^s\right)+\frac{1}{n-1}m_{st}\left(2T_{kr}^s+\beta_{kr}^s\right)\right]b^t,
  \end{align*}
  and it is true for any $b^t$, we obtain the identity
  \[
    \eta^{kr}m_{rs}\left(2T_{kt}^s+\beta_{kt}^s\right)-\frac{1}{n-1}\eta^{kr}m_{rt}\left(2T_{ks}^s+\beta_{ks}^s\right)+\frac{1}{n-1}\eta^{kr}m_{st}\left(2T_{kr}^s+\beta_{kr}^s\right)=0.
  \]
  Using Equation \eqref{eq:FirstEqFirstConstraintManifold} we can get rid of the last term, obtaining
  \[
    m_{rs}\left[\eta^{kr}\left(2T_{kt}^s+\beta_{kt}^s\right)-\frac{1}{n-1}\eta^{kr}\delta^s_t\left(2T_{kp}^p+\beta_{kp}^p\right)\right]=0.
  \]
  Because $m$ is an arbitrary antisymmetric matrix, it implies that
  \begin{multline}\label{eq:BetaTPrimary}
    \eta^{kr}\left(2T_{kt}^s+\beta_{kt}^s\right)+\eta^{ks}\left(2T_{kt}^r+\beta_{kt}^r\right)-\frac{1}{n-1}\left(\eta^{kr}\delta^s_t+\eta^{ks}\delta^r_t\right)\left(2T_{kp}^p+\beta_{kp}^p\right)=0;
  \end{multline}
  Making $r=t$ and summing over this index we obtain
  \[
    \eta^{ks}\left(2T_{kt}^t+\beta_{kt}^t\right)-\frac{1}{n-1}\eta^{ks}\left(n+1\right)\left(2T_{kp}^p+\beta_{kp}^p\right)=0,
  \]
  that is,
  \[
    2T_{kt}^t+\beta_{kt}^t=0.
  \]
  Also, multiplying each side of \eqref{eq:BetaTPrimary} by $\eta_{rs}$ and summing up every term, we obtain
  \[
    2\left(2T_{st}^s+\beta_{st}^s\right)-\frac{2}{n-1}\left(2T_{tp}^p+\beta_{tp}^p\right)=0.
  \]
  Therefore,
  \[
    T_{kt}^t=0=\beta_{kt}^t.\qedhere
  \]
\end{proof}
\begin{cor}
  When restricted to the primary constraint submanifold, the components of the torsion tensor are elements in $B$. Also, on this submanifold,
  \[
    M_{\mu\nu}^\sigma=\frac{1}{2}\left(\Gamma_{\mu\nu}^\sigma+\Gamma_{\nu\mu}^\sigma\right).
  \]
\end{cor}
\begin{proof}
  They are antisymmetric, and by the previous result, also traceless.
\end{proof}

\subsection{Hamiltonian form and Hamilton equations for gravity with basis}
\label{sec:hamilt-form-hamilt}

We have constructed a projection compatible with the Lagrangian form $\lambda_\cL$; according to the scheme developed previously, it will be possible to find a Hamiltonian form and so a set of Hamilton equations for GR with basis.

\subsubsection{The Hamiltonian form}
\label{sec:hamiltonian-form}

In order to proceed, let us introduce some notation; we will indicate with $L\subset GL\left(n\right)$ the Lorentz group defined by $\eta$. Then the $L$-invariance of the Lagrangian $\lambda_\cL$ allows us to prove the following result. Also, as always, let $d^nx$ be the canonical volume on $M$ associated with the coordinates $\left(x^\mu\right)$ and define
\[
  dx_{\mu_1\cdots\mu_k}:=\frac{\partial}{\partial x^{\mu_1}}\lrcorner\frac{\partial}{\partial x^{\mu_2}}\lrcorner\cdots\frac{\partial}{\partial x^{\mu_k}}\lrcorner d^nx.
\]

\begin{prop}[\cite{capriotti14:_differ_palat}]
  Let $\left(x^\mu,e_k^\nu,e^\sigma_{j\mu}\right)$ be the set of jet coordinates introduced above on $J^1\tau$. Then we can introduce a set of $L$-invariant functions $g^{\mu\nu}$ on $J^1\tau$ through
  \[
    {g}^{\mu\nu}:=\eta^{kl}e_k^\mu e_l^\nu.
  \]
  In terms of these functions the Lagrangian can be written as
  \begin{multline}
    {\lambda_{\cL}}=\sqrt{-\det{{g}}}{g}^{\kappa\phi}\left(\dif\Gamma^\gamma_{\rho\phi}\wedge\dif x^\rho+\Gamma^\sigma_{\delta\phi}\Gamma^\gamma_{\beta\sigma}\dif x^\beta\wedge\dif x^\delta\right)\wedge dx_{\gamma\kappa}+\\
    +\beta_{\mu\nu}^\rho\left(\dif {g}^{\mu\nu}+\left({g}^{\mu\sigma}{\Gamma}^\nu_{\gamma\sigma}+{g}^{\nu\sigma}{\Gamma}^\mu_{\gamma\sigma}\right)\dif x^\gamma\right)\wedge dx_\rho
  \end{multline}
\end{prop}

Let $i_0:W_0\hookrightarrow\widehat{W_\cL}$ be the canonical immersion and define
\[
  \lambda_0:=i_0^*\lambda_\cL;
\]
please note that
\[
  {\Gamma}^\sigma_{\mu\nu}=\frac{1}{2}\left({\Gamma}^\sigma_{\mu\nu}+{\Gamma}^\sigma_{\nu\mu}\right)+\frac{1}{2}\left({\Gamma}^\sigma_{\mu\nu}-{\Gamma}^\sigma_{\nu\mu}\right)=\frac{1}{2}\left({\Gamma}^\sigma_{\mu\nu}+{\Gamma}^\sigma_{\nu\mu}\right)+e^\sigma_ke^i_\mu e^j_\nu T_{ij}^k,
\]
and that on $W_0$ the matrix $T$ belongs to $B$. By the $B$-invariance of $\lambda_\cL$, it results that
\begin{multline*}
  \lambda_0=\frac{1}{2}\sqrt{-\det{{g}}}{g}^{\kappa\phi}\\
  \left[\dif\left(\Gamma^\gamma_{\rho\phi}+\Gamma^\gamma_{\phi\rho}\right)\wedge\dif x^\rho+\frac{1}{2}\left(\Gamma^\sigma_{\delta\phi}+\Gamma^\sigma_{\phi\delta}\right)\left(\Gamma^\gamma_{\beta\sigma}+\Gamma^\gamma_{\sigma\beta}\right)\dif x^\beta\wedge\dif x^\delta\right]\wedge dx_{\gamma\kappa}+\\
  +\beta_{\mu\nu}^\rho\left[\dif {g}^{\mu\nu}+\frac{1}{2}\left({g}^{\mu\sigma}\left({\Gamma}^\nu_{\gamma\sigma}+{\Gamma}^\nu_{\sigma\gamma}\right)+{g}^{\nu\sigma}\left({\Gamma}^\mu_{\gamma\sigma}+{\Gamma}^\mu_{\sigma\gamma}\right)\right)\dif x^\gamma\right]\wedge dx_\rho.
\end{multline*}
So there exists a $n$-form $\lambda_h$ on $W^\ddagger$, the \emph{Hamiltonian potential form}, which is given locally by
\begin{align*}
  \lambda_h=\sqrt{-\det{{g}}}{g}^{\kappa\phi}\left(\dif M^\gamma_{\rho\phi}\wedge\dif x^\rho+M^\sigma_{\delta\phi}M^\gamma_{\beta\sigma}\dif x^\beta\wedge\dif x^\delta\right)\wedge dx_{\gamma\kappa}+\\
  +\beta_{\mu\nu}^\rho\left[\dif {g}^{\mu\nu}+\left({g}^{\mu\sigma}{M}^\nu_{\gamma\sigma}+{g}^{\nu\sigma}{M}^\mu_{\gamma\sigma}\right)\dif x^\gamma\right]\wedge dx_\rho
\end{align*}
in the coordinates $\left(x^\mu,e^\mu_k,M_{\mu\nu}^\sigma,\beta_{pq}^r\right)$ on $W^\ddagger$, and such that
\[
  p_0^*\lambda_h=\lambda_0.
\]
Using the fact that
\[
  dx^\mu\wedge dx_\nu=\delta^\mu_\nu d^nx
\]
and
\[
  dx^\mu\wedge dx_{\nu\sigma}=\delta^\mu_\nu dx_\sigma-\delta^\mu_\sigma dx_\nu,
\]
we can simplify further the Hamiltonian potential form as
\begin{multline*}
  \lambda_h=\sqrt{-\det{{g}}}{g}^{\kappa\phi}\left[\dif M^\gamma_{\gamma\phi}\wedge dx_{\kappa}-\dif M^\gamma_{\kappa\phi}\wedge dx_{\gamma}+\left(M^\sigma_{\kappa\phi}M^\beta_{\beta\sigma}-M^\sigma_{\gamma\phi}M^\gamma_{\kappa\sigma}\right)d^nx\right]+\\
  +\beta_{\mu\nu}^\rho\left[\dif {g}^{\mu\nu}\wedge dx_\rho+2{g}^{\mu\sigma}{M}^\nu_{\rho\sigma}\dif^n x\right].
\end{multline*}
The Hamiltonian form is the differential of this form, namely
\begin{multline*}
  \Omega_h=\frac{1}{2}\sqrt{-\det{{g}}}{g}^{\kappa\phi}g_{\omega\eta}dg^{\omega\eta}\wedge\\
  \wedge\left[\dif M^\gamma_{\gamma\phi}\wedge dx_{\kappa}-\dif M^\gamma_{\kappa\phi}\wedge dx_{\gamma}-\left(M^\sigma_{\kappa\phi}M^\beta_{\beta\sigma}-M^\sigma_{\gamma\phi}M^\gamma_{\kappa\sigma}\right)d^nx\right]+\\
  +\sqrt{-\det{{g}}}d{g}^{\kappa\phi}\wedge\left[\dif M^\gamma_{\gamma\phi}\wedge dx_{\kappa}-\dif M^\gamma_{\kappa\phi}\wedge dx_{\gamma}-\left(M^\sigma_{\kappa\phi}M^\beta_{\beta\sigma}-M^\sigma_{\gamma\phi}M^\gamma_{\kappa\sigma}\right)d^nx\right]-\\
  -\sqrt{-\det{{g}}}{g}^{\kappa\phi}\left(M^\sigma_{\kappa\phi}dM^\beta_{\beta\sigma}+M^\beta_{\beta\sigma}dM^\sigma_{\kappa\phi}-M^\sigma_{\gamma\phi}dM^\gamma_{\kappa\sigma}-M^\gamma_{\kappa\sigma}dM^\sigma_{\gamma\phi}\right)\wedge d^nx+\\
  +d\beta_{\mu\nu}^\rho\wedge\left[\dif {g}^{\mu\nu}\wedge dx_\rho+2{g}^{\mu\sigma}{M}^\nu_{\rho\sigma}\dif^n x\right]+2\beta_{\mu\nu}^\rho{g}^{\mu\sigma}d{M}^\nu_{\rho\sigma}\wedge\dif^n x+2\beta_{\mu\nu}^\rho{M}^\nu_{\rho\sigma}d{g}^{\mu\sigma}\wedge\dif^n x.
\end{multline*}
Because the Ricci forms
\[
  R_{\phi\kappa}:=d M^\gamma_{\gamma\phi}\wedge dx_{\kappa}-d M^\gamma_{\kappa\phi}\wedge dx_{\gamma}-\left(M^\sigma_{\kappa\phi}M^\beta_{\beta\sigma}-M^\sigma_{\gamma\phi}M^\gamma_{\kappa\sigma}\right)d^nx
\]
are independent of the metric, it is convenient to have at our disposal this version for the Hamiltonian form
\begin{multline}
  \Omega_h=\frac{1}{2}\sqrt{-\det{{g}}}\left({g}^{\kappa\phi}g_{\omega\eta}dg^{\omega\eta}+2d{g}^{\kappa\phi}\right)\wedge R_{\phi\kappa}+\sqrt{-\det{{g}}}{g}^{\kappa\phi} dR_{\phi\kappa}+\\
  +d\beta_{\mu\nu}^\rho\wedge\left[\dif {g}^{\mu\nu}\wedge dx_\rho+2{g}^{\mu\sigma}{M}^\nu_{\rho\sigma}\dif^n x\right]+2\beta_{\mu\nu}^\rho{g}^{\mu\sigma}d{M}^\nu_{\rho\sigma}\wedge\dif^n x+2\beta_{\mu\nu}^\rho{M}^\nu_{\rho\sigma}d{g}^{\mu\sigma}\wedge\dif^n x.\label{eq:OmegaHEasy}
\end{multline}
There is another way to write this multisympletic form (usually called Hamilton-Cartan form) that makes it easy to compare it with analogous structures found in the literature \cite{0264-9381-32-9-095005,gaset19:_new_metric_affin_einst_palat}, namely
\begin{multline}
  \Omega_h=\\
  =dH\wedge d^nx+\frac{1}{2}\sqrt{-\det{{g}}}{g}^{\kappa\phi}g_{\omega\eta}dg^{\omega\eta}\wedge\left(\dif M^\gamma_{\gamma\phi}\wedge dx_{\kappa}-\dif M^\gamma_{\kappa\phi}\wedge dx_{\gamma}\right)+\\
  +\sqrt{-\det{{g}}}d{g}^{\kappa\phi}\wedge\left(\dif M^\gamma_{\gamma\phi}\wedge dx_{\kappa}-\dif M^\gamma_{\kappa\phi}\wedge dx_{\gamma}\right)+d\beta_{\mu\nu}^\rho\wedge\dif {g}^{\mu\nu}\wedge dx_\rho
\end{multline}
where
\[
  H:=\sqrt{-\det{{g}}}{g}^{\kappa\phi}\left(M^\sigma_{\kappa\phi}M^\beta_{\beta\sigma}-M^\sigma_{\gamma\phi}M^\gamma_{\kappa\sigma}\right)+2\beta^\sigma_{\mu\nu}g^{\mu\sigma}M_{\rho\sigma}^\nu
\]
is a local Hamiltonian function. This Hamiltonian form is similar to the one found by Gaset and Rom{\'a}n-Roy, but has an additional term associated to the constraints employed in the formulation of the underlying general variational problem.

\subsubsection{The Hamilton equations for GR with basis}
\label{sec:hamilt-equat-gr}

We are now ready to calculate the Hamilton equations (in the field theoretic sense) associated to the Hamiltonian form just found.

\begin{thm}\label{thm:HamiltonEqsGR}
  The Hamilton equations associated to $\Omega_h$ are the vacuum Einstein equations.
\end{thm}
\begin{proof}
  Let us calculate the equations of motion associated to this form. A local basis of vector fields can be constructed projecting $W_0$ along $p^\ddagger$; the equation defining this set is
  \[
    \eta^{kr}c_{rs}^t\left(2T_{kt}^r+\beta_{kt}^r\right)=0.
  \]
  Here we have used the underlying basis $e_k^\mu$ in order to turn greek indices into latin indices; using this convention, any vector $Y$ tangent in $p^\ddagger\left(w\right)$ to $p^\ddagger\left(W_0\right)$ can be written as
  \[
    Y=T_wp^\ddagger\left(X\right)
  \]
  for some vector $X$ tangent to $W_0$ in $w$. For $X$ of the form
  \[
    X=\tau_{kt}^r\frac{\partial}{\partial\Gamma_{kt}^r}+B_{kt}^r\frac{\partial}{\partial\beta_{kt}^r}
  \]
  the tangency condition becomes
  \begin{equation}\label{eq:TangencyCondition}
    \eta^{kr}c_{rs}^t\left[2\left(\tau_{kt}^r-\tau_{tk}^r\right)+B_{kt}^r\right]=0
  \end{equation}
  and the projection $Y$ takes the form
  \[
    Y=T_wp^\ddagger\left(X\right)=\frac{1}{2}\left(\tau_{kt}^r+\tau_{tk}^r\right)\frac{\partial}{\partial M_{kt}^r}+B_{kt}^r\frac{\partial}{\partial\beta_{kt}^r}.
  \]
  Therefore, a local basis for vector tangent to $p^\ddagger\left(W_0\right)$ is given by
  \[
    \left\{\frac{\partial}{\partial M^\zeta_{\alpha\xi}},B_{\mu\nu}^\rho\frac{\partial}{\partial\beta^\rho_{\mu\nu}},\frac{\partial}{\partial g^{\alpha\xi}},\xi_{W^\ddagger}:\xi\in\mathfrak{l}\right\}
  \]
  where $\mathfrak{l}$ is the Lie algebra of the Lorentz group $L$ and the $B$-coefficients obey the tangency condition \eqref{eq:TangencyCondition}; because of the $L$-invariance, no equation will rise from the contraction with these vectors, so the Hamilton equations will be consequence of the remaining ones. Also, keep in mind that because the previous discussion, there would be relationships between the vectors in the $\beta$-directions. Let us first contract in these directions; we obtain
  \[
    B^\rho_{\mu\nu}\frac{\partial}{\partial\beta^\rho_{\mu\nu}}\lrcorner\Omega_h=B^\rho_{\mu\nu}\left[\dif {g}^{\mu\nu}\wedge dx_\rho+2{g}^{\mu\sigma}{M}^\nu_{\rho\sigma}\dif^n x\right]
  \]
  where the $B$-coefficients are necessary because the $\beta$-variables are not independent on $p^\ddagger\left(W_0\right)\subset W^\ddagger$. Using Equation \eqref{eq:TangencyCondition} with $\tau_{kt}^r-\tau_{tk}^r=0$, it results that these coefficients should be traceless, and it yields to the consequence
  \begin{equation}\label{eq:WeylConnection}
    \dif {g}^{\mu\nu}\wedge dx_\rho+2{g}^{\mu\sigma}{M}^\nu_{\rho\sigma}\dif^n x=\frac{1}{n-1}g^{\mu\nu}\left(g_{\omega\eta}dg^{\omega\eta}\wedge dx_\rho+2M_{\beta\rho}^\beta d^nx\right).
  \end{equation}
  It means that the connection represented by the $M$-coordinates is a \emph{Weyl connection}, namely, a torsionless connection that is metric up to a trace \cite{10.1063/1.529582}. For the contraction with the $M$-directions we obtain
  \begin{multline*}
    \frac{\partial}{\partial M^\zeta_{\alpha\xi}}\lrcorner\Omega_h=\sqrt{-\det{g}}\left[dg^{\alpha\xi}\wedge dx_\zeta+\left(g^{\alpha\phi}M_{\zeta\phi}^\xi+g^{\kappa\xi}M_{\kappa\zeta}^\alpha\right)d^nx\right]-\\
    -\sqrt{-\det{g}}\delta^\alpha_\zeta\left[dg^{\kappa\xi}\wedge dx_\kappa+\left(g^{\kappa\phi}M_{\kappa\phi}^\xi+g^{\kappa\xi}M_{\kappa\beta}^\beta\right)d^nx\right]-\\
    -\frac{1}{2}\sqrt{-\det{g}}g^{\alpha\xi}\left(g_{\omega\eta}dg^{\omega\eta}\wedge dx_\zeta+2M_{\beta\zeta}^\beta d^nx\right)+\\
    +\frac{1}{2}\sqrt{-\det{g}}g^{\kappa\xi}\delta^\alpha_\zeta\left(g_{\omega\eta}dg^{\omega\eta}\wedge dx_\kappa+2M_{\beta\kappa}^\beta d^nx\right)+2g^{\mu\xi}\beta^\alpha_{\mu\zeta}
  \end{multline*}
  Taking into account Equation \eqref{eq:WeylConnection}, we obtain the following formula
  \begin{multline*}
    2g^{\mu\xi}\beta^\alpha_{\mu\zeta}=\frac{n-3}{n-1}\sqrt{-\det{g}}\Big[g^{\kappa\xi}\delta^\alpha_\zeta\left(g_{\omega\eta}dg^{\omega\eta}\wedge dx_\kappa+2M_{\beta\kappa}^\beta d^nx\right)-\\
    -g^{\alpha\xi}\left(g_{\omega\eta}dg^{\omega\eta}\wedge dx_\zeta+2M_{\beta\zeta}^\beta d^nx\right)\Big]
  \end{multline*}
  for the $\beta$-variables. Let us suppose now that $n\not=3$. Summing up $\alpha$ with $\zeta$ in the left hand side we obtain $0$; it gives us the equation
  \[
    g^{\kappa\xi}\left(n-1\right)\left(g_{\omega\eta}dg^{\omega\eta}\wedge dx_\kappa+2M_{\beta\kappa}^\beta d^nx\right)=0,
  \]
  and so every $\beta$ should annihilate. The same is true for the Weyl form
  \[
    Q_\kappa:=g_{\omega\eta}dg^{\omega\eta}\wedge dx_\kappa+2M_{\beta\kappa}^\beta d^nx,
  \]
  meaning that the connection represented by the $M$-variables is metric. The equation defining the primary constraint submanifold then reduces to
  \[
    \eta^{kr}c_{rs}^tT^s_{kt}=0.
  \]
  Using Lemma \ref{lem:SolCompatibility}, we obtain from here that
  \begin{equation}\label{eq:AntiSymFixedTrace}
    \eta^{kr}a_{rs}^tT^s_{kt}=0
  \end{equation}
  where $a_{rs}^t$ is antisymmetric in the lower indices and has fixed trace $a_{qr}^r=(n-1)c_{rq}^r/2$. Now, for any vector $a_s$ we have that
  \[
    \eta^{kr}\left(a_r\delta_s^t-a_s\delta_r^t\right)T^s_{kt}=0,
  \]
  so that Equation \eqref{eq:AntiSymFixedTrace} holds for any antisymmetric matrix (in the lower indices), without the restriction on the trace. Therefore we must have $T_{kt}^s=0$ on solutions of the Hamilton equations; that is, the connection that is solution of these equations should be also torsionless. Using formula \eqref{eq:OmegaHEasy} and the fact that $\beta=0$, the contraction of $\Omega_h$ with the $g$-variables yields to the Einstein equations in vacuum, as required.
\end{proof}
\begin{note}
  The case $n=3$ is very interesting. We still have that $\beta=0$ and so the contraction with elements in the $g$-directions also gives rise to the Einstein equations, but the connection does not reduce to a metric one; instead, it remains a connection of Weyl type. Additionally, it should be noted that this number depends on the signature of $\eta$, which we have chosen here to be $\left(n-1,1\right)$. For a general signature $\left(n-k,k\right)$, the same feature will happen when $n=k+2$. It seems to point out to an additional projective symmetry for the Hamilton equations in these dimensions.
\end{note}

\section{Towards a constraint algorithm for general variational problems}
\label{sec:constr-algor-sing}

It is well known that when working with classical field theory and the Lagrangian of the theory is singular, a constraint algorithm is needed in order to relate the solutions of the Hamilton equations with the solutions of the original Lagrangian field theory. We will try to reproduce this procedure in the general case, with the original field theory given by a triple $\left(\pi:W\to M,\lambda,\cI\right)$; this generality has two reasons: On the one hand, it will allow us to apply the algorithm to some interesting examples, and on the other hand, as a way to show that some constructions that are central in the case of classical variational problems loss their special status when are considered from this generalized viewpoint. One of the main challenges of this generalization is the fact that we have no equivalent notion here of the concept of ``singular Lagrangian''; because of this fact, we will call a general variational problem and a projection $p^\ddagger$ \emph{regular} if a procedure is known in order to lift any solution of the Hamilton equations to a solution of the unified problem. For example, for classical field theory the inverse of the Legendre transformation provides us with such a procedure. Whenever this lifting problem is impossible in some regions of the phase space $C$, we will say that the problem (together with the projection used to define the Hamiltonian theory!) is \emph{singular}.

Please note that, up to now, the subbundle $C$ involved in Definition \ref{def:hamilt-equat-gener} is arbitrary, as long as we are only interested in the Hamilton equations; when we try to relate its solutions with the solutions of the underlying unified variational problem, some relationships arise between this submanifold and the set $P_0\subset W_\lambda$. This section explores these relations.

\subsection{Final constraint submanifold and Hamilton equations}
\label{sec:final-constr-subm}

Let us consider the basic setting of Section \ref{sec:hamilt-equat-gener}; therefore, we have a map $p_\lambda:W_\lambda\to W^\ddagger$, we will admit that this map is horizontal respect to the form $\Theta_\lambda$, and we will single out a subbundle $P\subset P_0$ projecting via $p_\lambda$ onto the submanifold $C\subset W^\ddagger$. Given a solution for Hamilton equations on $C$, it should be interesting to find conditions under which this solution can be lifted to a solution for the unified problem on $W_\lambda$ posed by the form $\Theta_\lambda$. With this objective in mind, let us consider the following property.

\begin{defc}
  Given $\psi:U\subset M\to C$ solution for the Hamilton equations on $C$ and $V\subset T_PW_\lambda$ complementary to $TP+\ker{Tp_\lambda}$, a \emph{$V$-admissible lift for $\psi$} is a section $\widehat{\psi}:U\to P$ such that
  \begin{enumerate}
  \item $p_\lambda\circ\widehat{\psi}=\psi$, and
  \item $\widehat{\psi}^*\left(Z\lrcorner d\Theta_\lambda\right)=0$ for all $Z\in\mathfrak{X}\left(W_\lambda\right)$ such that
    \[
      Z\left(\rho\right)\in\left.V\right|_\rho+\ker{T_\rho p_\lambda}
    \]
    for every $\rho\in P$.
  \end{enumerate}
\end{defc}

Let us introduce the following definition.
\begin{defc}[Final constraint submanifold]\label{def:FinalConstraint}
  A subbundle $P_f\subset P_0$ projecting via $p_\lambda$ onto a submanifold $C_f:=p_\lambda\left(P_f\right)\subset W^\ddagger$ is a \emph{final constraint submanifold for $p_\lambda$} if and only if $C_f$ is final for $\Omega_h$ and such that every solution $\psi:U\subset M\to C_f$ of the Hamilton equations on $C_f$ admits a $V_\psi$-admissible lift, for some complement $V_\psi$ of $TP_f+\mathop{\text{ker}}{Tp_\lambda}$ in $T_{P_f}W_\lambda$.
\end{defc}
The relevance of this definition relies in the following property regarding the lifting of solutions of the Hamilton equations to the unified problem.
\begin{thm}
  Let $p_f:P_f\to C_f$ be a final constraint submanifold; let $\psi:U\to C_f$ be a solution for the Hamilton equations on $C_f$. Then $\widehat{\psi}:U\to P_f$ is a solution for the unified variational problem on $W_\lambda$. Conversely, if $\widehat{\psi}:U\to P_f\subset W_\lambda$ is a solution of the unified variational problem on $W_\lambda$ taking values in a final constraint submanifold, then $\psi:=p\circ\widehat{\psi}:U\to C_f$ is a solution for the associated Hamilton equations on $C_f$.
\end{thm}
\begin{proof}
  Any $p_\lambda$-projectable vector field $X$ on $P_f\subset W_\lambda$, when restricted to $P_f$, can be written as
  \[
    X\circ i_f=Ti_f\circ X_f+Z\circ i_f+W
  \]
  where $X_f\in\mathfrak{X}\left(P_f\right)$ projects onto a vector field $\overline{X}_f\in\mathfrak{X}\left(C_f\right)$, $Z\in\mathfrak{X}^{Vp_\lambda}\left(W_\lambda\right)$ and $W\left(\rho\right)\in V_\psi$ for every $\rho\in P_f$. Then, using Lemma \ref{lem:HamiltonFormProp}
  \begin{align*}
    \widehat{\psi}^*\left(X\lrcorner d\Theta_\lambda\right)&=\widehat{\psi}^*\left(\left(Ti_f\circ X_f\right)\lrcorner d\Theta_\lambda\right)+\widehat{\psi}^*\left(\left(Z+W\right)\lrcorner d\Theta_\lambda\right)\\                                                      &=\widehat{\psi}^*\left(X_f\lrcorner\left(i_f^*\Omega_\lambda\right)\right)\\ &=\widehat{\psi}^*\left(X_f\lrcorner\left(p_\lambda^*\Omega_h\right)\right)\\
                                                           &=\left(p_\lambda\circ\widehat{\psi}\right)^*\left[\left(Tp_\lambda\circ X_f\right)\lrcorner\Omega_h\right]\\                                                           &=\psi^*\left(\overline{X}_f\lrcorner\Omega_h\right)\\
                                                           &=0
  \end{align*}
  as required.
\end{proof}
\begin{note}
  There is a quite important case not covered by this theorem: Are there solutions of the unified variational problem not taking values in a final constraint submanifold? 
\end{note}
\begin{example}[The final constraint submanifold for the classical variational problem associated to a regular Lagrangian density]
  In this case we have that the projection map is $p_\lambda:=p_\cL^\ddagger$. Then $P_f=P_0,C_f=J^{k+1}\pi^\ddagger$, and the section $\gamma_0$ is implicitly defined by the equations
  \[
    p_\alpha^{\left(I,i\right)}=\frac{\partial L}{\partial u_{I+1_i}}.
  \]
  Namely, to give this section is to provide a set of functions $u_{I}^\alpha,\abs{I}=k+1,$ on $J^{k+1}\pi^\ddagger$ with the additional requirement that $\mathop{\text{Im}}\gamma_0\subset P_0$; because of the regularity of $\cL$, we can use the equations defining $P_0$ to find these functions. 
\end{example}

\subsection{The constraint algorithm for general variational problems}
\label{sec:constr-algor-gener}

We are now ready to formulate a version of the constraint algorithm working for general variational problems. Take $P_0$ as initial submanifold, $C_0:=p_\lambda\left(P_0\right)$ and define for every $l=0,1,2,\cdots$ the set $C_{l+1}\subset C_l$ subject to the following requirement: $C_{l+1}$ is  the maximal set in $C_l'$ (the solution set of $C_l$) such that there exists a complement $V_l$ of $\mathop{\text{ker}}{Tp_\lambda}+TP_l$ in $T_{P_l}W_\lambda$ with the property that every $\psi\in H_{C_l}$ taking values in $C_{l+1}$ has a $V_l$-admissible lift with image in $P_l$. Then define $P_{l+1}:=p_\lambda^{-1}\left(C_{l+1}\right)$.
\begin{prop}
  Let $L$ be the minimum integer such that $C_L=C_{L+1}$. Then $P_L$ is a final constraint submanifold for $p_\lambda$. 
\end{prop}
\begin{proof}
  When $C_{L+1}=C_L$, then $C_L$ is final for $\Theta_h$ and by Definition \ref{def:FinalConstraint}, $P_L$ is a final constraint submanifold.
\end{proof}

\begin{note}
  It could be possible for a map to have as final constraint submanifold the empty set.
\end{note}

\subsection{Examples}
\label{sec:examples-1}

Let us apply the previous scheme to some well-known systems, in order to show how it allows us to extract the constraints associated to a singular Lagrangian.

\subsubsection{A pair of hanging springs}
\label{sec:pair-hanging-springs}

Let us consider an example from \cite{brown2023singular}. It represents two springs (with spring constants $k_1$ and $k_2$) attached end-to-end and hanging from the roof; additionally, a mass $m$ is located at the free end of the composed spring. If we indicate with $x_1$ and $x_2$ the length of each of the springs, the Lagrangian of the system becomes
\[
  L\left(x,\dot{x}\right):=\frac{m}{2}\left(\dot{x}^1+\dot{x}^2\right)^2+mg\left(x^1+x^2\right)-\frac{k_1}{2}\left(x^1-l_1\right)^2-\frac{k_2}{2}\left(x^2-l_2\right)^2,
\]
where $l_1$ and $l_2$ are the equilibrium length for each of the springs. Thus we are dealing with a classical variational problem associated to the trivial bundle $\text{pr}_1:\mR\times\mR^2\to\mR$; therefore we need to consider the Lagrangian $L$ as defining a Lagrangian density $\cL=Ldt$ on $J^1\text{pr}_1=\mR\times T\mR^2$ with coordinates $\left(t,x^1,x^2,\dot{x}^1,\dot{x}^2\right)$. The space of forms $W_L$ has global coordinates $\left(t,x^1,x^2,\dot{x}^1,\dot{x}^2,p_1,p_2\right)$ in such a way that $\rho\in W_L$ if and only if
\[
  \rho=L\left(x,\dot{x}\right)dt+p_i\left(dx^i-\dot{x}^idt\right).
\]
The projection defining the Hamiltonian theory is in this case
\[
  p_\cL^\ddagger:W_L\to\mR\times\mR^2\times\mR^2:\left(t,x^1,x^2,\dot{x}^1,\dot{x}^2,p_1,p_2\right)\mapsto\left(t,x^1,x^2,p_1,p_2\right).
\]
Then
\[
  Vp_\cL^\ddagger=\left<\frac{\partial}{\partial\dot{x}^1},\frac{\partial}{\partial\dot{x}^2}\right>
\]
and because $\Theta_L=L\left(x,\dot{x}\right)dt+p_i\left(dx^i-\dot{x}^idt\right)$, it follows that $P_0$ will be defined by the equations
\[
  \left[m\left(\dot{x}^1+\dot{x}^2\right)-p_1\right]dt=0=\left[m\left(\dot{x}^1+\dot{x}^2\right)-p_2\right]dt,
\]
namely,
\[
  P_0=\left\{\left(t,x^1,x^2,\dot{x}^1,\dot{x}^2,p_1,p_2\right):p_1=p_2=m\left(\dot{x}^1+\dot{x}^2\right)\right\}.
\]
This submanifold can be parameterized through the coordinates $\left(t,x^1,x^2,\dot{x}^1,p\right)$ and we can embed it into $W_L$ with the map
\[
  i_0\left(t,x^1,x^2,\dot{x}^1,p\right)=\left(t,x^1,x^2,\dot{x}^1,\frac{p}{m}-\dot{x}^1,p,p\right).
\]
The Hamiltonian form $\Theta_0$ should be retrieved from the pullback form
\begin{align*}
  \Theta_0&=i_0^*\Theta_L\\
  &=L\left(x,\dot{x}\right)dt+p\left(dx^1-\dot{x}^1dt+dx^2-\dot{x}^2dt\right)\\
  &=\left[L\left(x,\dot{x}\right)-p\left(\dot{x}^1+\dot{x}^2\right)\right]dt+p\left(dx^1+dx^2\right);
\end{align*}
now, because
\begin{align}
  L\left(x,\dot{x}\right)&-p\left(\dot{x}^1+\dot{x}^2\right)=\cr
  &=\frac{m}{2}\left(\frac{p}{m}\right)^2+mg\left(x^1+x^2\right)-\frac{k_1}{2}\left(x^1-l_1\right)^2-\frac{k_2}{2}\left(x^2-l_2\right)^2-\frac{p^2}{m}\cr
  &=mg\left(x^1+x^2\right)-\frac{k_1}{2}\left(x^1-l_1\right)^2-\frac{k_2}{2}\left(x^2-l_2\right)^2-\frac{p^2}{2m}\cr
  &=:-H_0\left(t,x,p\right),\label{eq:FirstHamiltonian}
\end{align}
we obtain
\[
  \Theta_0=-H_0\left(t,x,p\right)dt+p\left(dx^1+dx^2\right).
\]
The Hamilton-Cartan equations on $C_0=\left\{\left(t,x^1,x^2,p\right)\right\}$ become
\begin{align*}
  &dp-\frac{\partial H_0}{\partial x^1}dt=0\\
  &dp-\frac{\partial H_0}{\partial x^2}dt=0\\
  &\left(dx^1+dx^2\right)-\frac{\partial H_0}{\partial p}dt=0;
\end{align*}
using Equation \eqref{eq:FirstHamiltonian} it yields to 
\begin{align}
  &dp+\left[mg-k_1\left(x^1-l_1\right)\right]dt=0\label{eq:FirstHamilton}\\
  &dp+\left[mg-k_2\left(x^2-l_2\right)\right]dt=0\label{eq:SecondHamilton}\\
  &\left(dx^1+dx^2\right)-\frac{p}{m}dt=0.\label{eq:LastHamilton}
\end{align}
Therefore, $C_0$ is not a final manifold for $\Theta_h$, because any solution should live in the submanifold $C_1\subset C_0$ given by the equation
\[
  k_1\left(x^1-l_1\right)-k_2\left(x^2-l_2\right)=0
\]
obtained from the subtraction of the second to the first equation. Given that
\[
  \dot{x}^2=\frac{p}{m}-\dot{x}^1,
\]
it results that $TP_0=U_0^0$ and $\mathop{\text{ker}}{Tp_L}=W_0^0$, where
\begin{align*}
  U_0&=\left<d\dot{x}^2-\frac{1}{m}dp+d\dot{x}^1,dp_1-dp_2\right>\\
  W_0&=\left<dx^1,dx^2,dp_1,dp_2\right>.
\end{align*}
Then, because
\[
  U_0+W_0=\left<dx^1,dx^2,dp_1,dp_2,d\dot{x}^2-\frac{1}{m}dp+d\dot{x}^1\right>,
\]
we have that
\[
  TP_0\cap\mathop{\text{ker}}{Tp_L}=\left(U_0+W_0\right)^0=\left<dx^1,dx^2,dp_1,dp_2,d\dot{x}^2-\frac{1}{m}dp+d\dot{x}^1\right>^0,
\]
and so the complement $V_0$ should have dimension $1$. Using as complement the subbundle
\[
  V_0:=\left<\frac{\partial}{\partial p_1}\right>,
\]
we obtain that the lifting $\widehat{\psi}$ for any solution $\psi$ of the Hamilton equations should annihilate the form
\[
  \alpha_0:=dx^1-\dot{x}^1dt.
\]
It means that if $\psi\left(t\right)=\left(t,x^1\left(t\right),x^2\left(t\right),p\left(t\right)\right)$ and $\widehat{\psi}\left(t\right)=\left(t,x^1\left(t\right),x^2\left(t\right),\dot{x}^1\left(t\right),p\left(t\right)\right)$, then the function $\dot{x}^1=\dot{x}^1\left(t\right)$ is determined by the formula
\begin{equation}\label{eq:SolutionLift}
  \dot{x}^1\left(t\right)=\frac{\text{d}x^1}{\text{d}t}\left(t\right).
\end{equation}
Next, and according to the algorithm described in Section \ref{sec:constr-algor-gener}, we must define $P_1:=\left(p_L^\ddagger\right)^{-1}\left(C_1\right)$, and consider the Hamilton equations on $C_1$; we obtain
\begin{align*}
  \Theta_1&:=i_1^*\Theta_0\\
  &=-H_1\left(t,x_1,p\right)dt+\frac{\left(k_1+k_2\right)p}{k_2}dx^1
\end{align*}
where
\[
  H_1\left(t,x_1,p\right)=-\frac{mg}{k_2}\left[\left(k_1+k_2\right)x_1+l_2k_2-k_1l_1\right]+\frac{k_1\left(k_1+k_2\right)}{2k_2}\left(x^1-l_1\right)^2+\frac{p^2}{2m}.
\]
It is nothing but the Hamiltonian given by Equation $\left(15\right)$ in \cite{brown2023singular}. The Hamilton equations on $C_1$ become
\[
  -\left[\frac{mg\left(k_1+k_2\right)}{k_2}-\frac{k_1\left(k_1+k_2\right)}{2k_2}\left(x^1-l_1\right)\right]dt+\frac{k_1+k_2}{k_2}dp=0
\]
by contracting with $\partial/\partial x^1$, and
\[
  \frac{p}{m}dt-\frac{k_1+k_2}{k_2}dx^1=0
\]
when contracting with $\partial/\partial p$; because no restriction is necessary to find solutions to these equations, we can conclude that $C_1$ is a final submanifold for $\Omega_h$.

Now, it is necessary to prove that for every solution $\psi\left(t\right)=\left(t,x^1\left(t\right),p\left(t\right)\right)$ of the above equations we can construct a section $\widehat{\psi}\left(t\right)=\left(t,x^1\left(t\right),\dot{x}^1\left(t\right),p\left(t\right)\right)$ of $P_1\to\mR$ covering $\psi$, and such that
\[
  \widehat{\psi}^*\left(Z\lrcorner d\Theta_L\right)=0,\qquad Z\in V_1
\]
for $V_1$ a complement of $TP_1+\mathop{\text{ker}}{Tp_L^\ddagger}\subset T_{P_1}\left(W_L\right)$.
Because $TP_1=U_1^0$, where
\[
  U_1=U_0+\left<k_1dx^1-k_2dx^2\right>,
\]
it results that
\[
  U_1+W_0=U_0+W_0,
\]
and so $V_1=V_0$, namely, $\widehat{\psi}$ can be constructed using Equation \eqref{eq:SolutionLift} as before. Then $C_1$ is final for $\Theta_h$, and also $P_1=P_f$, a final constraint submanifold.

\subsubsection{Wave equation}
\label{sec:wave-equation}

Let us consider how can we get a non standard Hamiltonian theory for the wave equation. In order to proceed, consider the bundle
\[
  \text{pr}_{12}:\mR^3\to\mR^2:\left(t,x,u\right)\mapsto\left(t,x\right);
\]
then the wave equation can be represented by pulling back the contact structure on $J^2\text{pr}_{12}$ to the submanifold
\[
  R:=\left\{\left(t,x,u,u_t,u_x,u_{tt},u_{tx},u_{xx}\right)\in J^2\text{pr}_{12}:u_{tt}=u_{xx}\right\}\subset J^2\text{pr}_{12}.
\]
Define $\pi:R\to\mR^2$ as the restriction of $\left(\text{pr}_{12}\right)_{02}$ to $R$; then a general variational problem representing this equation could be
\[
  \left(\pi:R\to\mR^2,0,\cI_R\right)
\]
where $\cI_R$ is the pullback of the contact structure on $J^2\text{pr}_{12}$ to $R$. Using Gotay construction, we define the bundle
\[
  W:=I_R^m\subset\wedge^m_2\left(J^2\text{pr}_{12}\right)
\]
and an element $\rho\in W$ if and only if
\[
  \rho=p^1\theta\wedge dx-p^2\theta\wedge dt+p^{11}\theta_1\wedge dx-p^{12}\theta_1\wedge dt+p^{21}\theta_2\wedge dx-p^{22}\theta_2\wedge dt, 
\]
where
\[
  \theta=du-u_tdt-u_xdx,\quad\theta_1=du_t-u_{tt}dt-u_{tx}dx,\quad\theta_2=du_x-u_{tx}dt-u_{tt}dx
\]
are the canonical generators for the contact structure on $R$ (where $u_{tt}=u_{xx}$). The canonical form $\Theta$ on $W$ is thus given by the formula
\begin{multline*}
  \Theta=\\
  =-\left[p^1u_t+p^2u_x+\left(p^{11}+p^{22}\right)u_{tt}+\left(p^{12}+p^{21}\right)u_{tx}\right]dt\wedge dx+p^1du\wedge dx-p^2du\wedge dt+\\
  +p^{11}du_t\wedge dx-p^{12}du_t\wedge dt+p^{21}du_x\wedge dx-p^{22}du_x\wedge dt.
\end{multline*}
We will use the following projection onto $W^\ddagger:=\mR^{10}$
\[
  p:W\to W^\ddagger:\left(t,x,u,u_t,u_x,u_{tt},u_{tx},p^I\right)\mapsto\left(t,x,u,u_t,u_x,p^I\right),\qquad\abs{I}\leq 2,
\]
in order to construct the Hamilton equations. Then
\[
  Vp=\left<\frac{\partial}{\partial u_{tt}},\frac{\partial}{\partial u_{tx}}\right>,
\]
and the first constraint manifold $P_0\subset W$ becomes described by the equations
\[
  p^{11}+p^{22}=0,\qquad p^{12}+p^{21}=0.
\]
It means that
\[
  Vp=\left<dt,dx,du,du_t,du_x,dp^I\right>^0,\qquad TP_0=\left<dp^{11}+dp^{22},dp^{12}+dp^{21}\right>^0,
\]
and so $Vp\subset TP_0$; then for the complementary subbundle
\begin{equation}\label{eq:ComplementaryWave}
  V_0=\left<\frac{\partial}{\partial p^{11}},\frac{\partial}{\partial p^{12}}\right>
\end{equation}
we have that
\[
  T_{P_0}W=TP_0\oplus V_0.
\]
The Hamiltonian form on
\[
  C_0:=p\left(P_0\right)=\left\{\left(t,x,u,u_t,u_x,p^1,p^2,p^{11},p^{12},-p^{12},-p^{11}\right)\right\}
\]
becomes
\begin{multline*}
  \Theta_h=-\left(p^1u_t+p^2u_x\right)dt\wedge dx+p^1du\wedge dx-p^2du\wedge dt+\\
  +p^{11}\left(du_t\wedge dx+du_x\wedge dt\right)-p^{12}\left(du_t\wedge dt+du_x\wedge dx\right).
\end{multline*}
The Hamilton equations on $C_0$ are
\begin{align*}
  \left(\frac{\partial}{\partial u}\right)&:dp^1\wedge dx-dp^2\wedge dt=0,\\
  \left(\frac{\partial}{\partial u_t}\right)&:-p^1 dt\wedge dx-dp^{11}\wedge dx+dp^{12}\wedge dt=0,\\
  \left(\frac{\partial}{\partial u_x}\right)&:-p^2 dt\wedge dx-dp^{11}\wedge dt+dp^{12}\wedge dx=0,\\
  \left(\frac{\partial}{\partial p^1}\right)&:-u_t dt\wedge dx+du\wedge dx=0,\\
  \left(\frac{\partial}{\partial p^2}\right)&:-u_x dt\wedge dx-du\wedge dt=0,\\
  \left(\frac{\partial}{\partial p^{11}}\right)&:du_t\wedge dx+du_x\wedge dt=0,\\
  \left(\frac{\partial}{\partial p^{12}}\right)&:du_t\wedge dt+du_x\wedge dx=0.
\end{align*}
No further restrictions are necessary in order to ensure that this system has solutions; therefore, $C_0$ is final for $\Theta_h$. Thus, the next step in the constraint algorithm is to construct a lift $\widehat{\psi}:U\subset\mR^2\to P_0$ for every solution $\psi:U\to C_0$ of these equations, such that
\[
  \widehat{\psi}^*\left(Z\lrcorner d\Theta\right)=0
\]
for every $Z\in\mathfrak{X}^{V_0}\left(W\right)$; using Equation \eqref{eq:ComplementaryWave}, it follows that
\[
  \widehat{\psi}=\left(t,x,,u,u_t,u_x,u_{tt},u_{tx},p^1,p^2,p^{11},p^{12},-p^{12},-p^{11}\right)
\]
must be integral for the ideal generated by the forms
\[
  -u_{tt}dt\wedge dx+du_t\wedge dx,\qquad-u_{tx}dt\wedge dx+du_t\wedge dt.
\]
It is straightforward to check that we can use the equations associated to these forms to choose the unknown functions $u_{tt},u_{tx}$, and also the lift is guaranteed; so, $P_0$ is a final constraint submanifold.

\printbibliography

\end{document}